\documentclass[%
 reprint,
 amsmath,amssymb,
 aps,floatfix,superscriptaddress,
pra,
nofootinbib]
{revtex4-2}

\usepackage{graphicx}
\usepackage{dcolumn}
\usepackage{bm}
\usepackage{epstopdf}
 \usepackage[usenames,dvipsnames]{pstricks}
 \usepackage{epsfig}
 \usepackage{pst-grad} 
 \usepackage{pst-plot} 
 \usepackage[space]{grffile} 
 \usepackage{etoolbox} 
 \makeatletter 
 \patchcmd\Gread@eps{\@inputcheck#1 }{\@inputcheck"#1"\relax}{}{}
 \makeatother
\usepackage{booktabs}
\usepackage{boxhandler}
\usepackage{qcircuit}
\usepackage{subcaption}
\usepackage{braket}
\usepackage[cmintegrals]{newtxmath}
\usepackage{multirow}
\usepackage{enumitem}

\usepackage{amsmath,amssymb,amsfonts,amsthm}

\usepackage{algorithmic}
\makeatletter
\let\MYcaption\@makecaption
\makeatother

\usepackage[font=footnotesize]{subcaption}
\newtheorem{thm}{Theorem}

\newtheorem{definition}{Definition}
\makeatletter
\let\@makecaption\MYcaption
\makeatother
\usepackage{url}
\usepackage{braket}
\usepackage{textcomp}
\usepackage{multirow}
\usepackage{xcolor}

\begin{document}


\title{Multi-qubit time-varying quantum channels for NISQ-era superconducting quantum processors}

\author{Josu Etxezarreta Martinez} 
\email[Correspondence: ]{jetxezarreta@tecnun.es}
    \affiliation{Department of Basic Sciences, Tecnun - University of Navarra, 20018 San Sebastian, Spain}

\author{Patricio Fuentes $^{1,\dagger}$, Antonio deMarti iOlius} 
\email[Correspondence: ]{\{pfuentesu, ademartio\}@tecnun.es}
    \thanks{These two authors contributed equally to this work.}
    \affiliation{Department of Basic Sciences, Tecnun - University of Navarra, 20018 San Sebastian, Spain}

\author{Javier Garcia-Frias}
        \email[Correspondence: ]{jgf@udel.edu}
    \affiliation{Department of Electrical and Computer Engineering, University of Delaware, Newark, 19716 Delaware, USA}
    
\author{Javier Rodr\'iguez Fonollosa}
        \email[Correspondence: ]{javier.fonollosa@upc.edu}
    \affiliation{Department de Teoria del Senyal i Comunicacions, Universitat Politecnica de Catalunya, 08034 Barcelona, Spain}

\author{Pedro M. Crespo}
        \email[Correspondence: ]{pcrespo@tecnun.es}
    \affiliation{Department of Basic Sciences, Tecnun - University of Navarra, 20018 San Sebastian, Spain}

\date{\today}

\begin{abstract}
Recent experimental studies have shown that the relaxation time ($T_1$) and the dephasing time ($T_2$) of superconducting qubits fluctuate considerably over time. To appropriately consider this time-varying nature of the $T_1$ and $T_2$ parameters, a new class of quantum channels, known as Time-Varying Quantum Channels (TVQCs), has been proposed. In previous works, realizations of multi-qubit TVQCs have been assumed to be equal for all the qubits of an error correction block, implying that the random variables that describe the fluctuations of $T_1$ and $T_2$ are block-to-block uncorrelated, but qubit-wise perfectly correlated for the same block. Physically, the fluctuations of these decoherence parameters are explained by the incoherent coupling of the qubits with unstable near-resonant two-level-systems (TLS), which indicates that such variations may be local to each of the qubits of the system. This qubit-wise uncorrelated behaviour has been proven in experiments conducted on a two-qubit system constructed with superconducting technology. In this article, we perform a correlation analysis of the fluctuations of the relaxation times of multi-qubit quantum processors ibmq\_quito, ibmq\_belem, ibmq\_lima, ibmq\_santiago and ibmq\_bogota. Our results show that it is reasonable to assume that the fluctuations of the relaxation and dephasing times of superconducting qubits are local to each of the qubits of the system. Based on these results, we discuss the multi-qubit TVQCs when the fluctuations of the decoherence parameters for an error correction block are qubit-wise uncorrelated (as well as from block-to-block), a scenario we have named the Fast Time-Varying Quantum Channel (FTVQC). Furthermore, we lower bound the quantum capacity of general FTVQCs based on a quantity we refer to as the ergodic quantum capacity and we prove that they both coincide for the class of fast time-varying amplitude damping channels. Finally, we use numerical simulations to study the performance of quantum error correction codes (QECC) when they operate over FTVQCs. Our simulation outcomes show that QECC performance is worse over FTVQCs than over channels that are static in time, but this degradation is less severe than if the fluctuations are considered to be qubit-wise fully correlated.
\end{abstract}
\maketitle

\section{Introduction}
Fault-tolerant quantum computers have the potential to revolutionize the fields of computing and industry as we know them \cite{preskill}. However, the societal upheaval heralded by quantum computers will only be facilitated if these machines are equipped with Quantum Error Correction (QEC) strategies, the primary ingredient to make these devices fault-tolerant. The purpose of QEC methods, commonly known as QECCs, is to identify and correct the errors that inherently corrupt quantum information. These errors are inevitable given that they arise naturally as a result of the interaction of quantum mechanical systems with their surrounding environment. The different physical mechanisms through which quantum information can be corrupted are commonly amalgamated under the term known as decoherence \cite{josuchannels}. Understanding the way decoherence corrupts quantum information and mathematically modeling such effects is of paramount importance to construct QECCs that can successfully correct the errors that occur in real quantum hardware.

The quantum noise experienced by qubits constructed as two-level coherent quantum mechanical systems can be described accurately by considering two different phenomena: relaxation and pure dephasing. The effects of these mechanisms are quantified by the relaxation time, $T_1$, and the dephasing time, $T_2$. Thermal interactions between the quantum information system and the environment can be neglected if the temperature of the system is low enough, which is a valid premise for state-of-the-art superconducting quantum processors, as these devices are cooled down to temperatures in the order of millikelvins \cite{decoherenceBenchmarking,klimov,fluctAPS,fluctApp,SchlorPhD,wallraff}. Generally, the effects of decoherence on quantum information are modelled by means of quantum channels, which are completely-positive trace-preserving (CPTP) linear maps betwen spaces of operators \cite{josuchannels}. Most of the literature on QEC assumes that the quantum noise level experienced by the qubits in a particular system will be identical for each quantum information processing task, independently of when the task is performed \cite{bicycle,qldpc15,patrick,patrick2,QTC,EAQTC,josu,josu2,toric,QEClidar,surfaceFowler}. This means that the relaxation and pure dephasing times of the qubits are assumed to be fixed and time-invariant. However, this behaviour has been disproven in recent experimental studies on quantum processors \cite{decoherenceBenchmarking,klimov,fluctAPS,fluctApp,SchlorPhD,fluctErrorBars,fluctDeph}. In fact, these works have shown that $T_1$ and $T_2$ can experience time variations of up to $50\%$ of their mean value with coefficients of variation of approximately $25\%$. In order to correctly account for the time-varying nature of the decoherence defining parameters of superconducting qubits, researchers have come up with the class of time-varying quantum channels (TVQC), a mathematical abstraction that enables the inclusion of time fluctuations to the models of quantum noise \cite{TVQC,outage}.

The multi-qubit TVQCs in \cite{TVQC,outage} were constructed considering that the realizations of $T_1$ and $T_2$ would be identical for all the qubits in each particular error correction block but would change from block to block. This assumption implies that the random variables that describe the fluctuations of the decoherence parameters of each of the qubits are perfectly correlated for the same block, a scenario reminiscent of the classical block/slow fading channel. In reality, the time fluctuations of these decoherence parameters occur because of the incoherent coupling of an ensemble of environmentally unstable near-resonant two-level systems (TLS) with each particular qubit \cite{decoherenceBenchmarking,klimov,fluctAPS,SchlorPhD}. These TLS' actually emerge due to atomic-scale defects that are present in the Josephson Junctions (JJ) that are used to make up the superconducting qubits \cite{TLSdefects,TLSphysrevB}. In light of this, it is reasonable to assume that the origin of the time-dependent nature of $T_1$ and $T_2$ is local to each of the constituent qubits of the quantum machine, which ultimately means that the realizations of the random variables that describe the fluctuations of those parameters in the TVQC model may be considered qubit-wise uncorrelated for each particular QEC block. This intuition has been confirmed by experiments conducted for a $2$-qubit superconducting system (See Supplementary Note D of \cite{decoherenceBenchmarking}), where the authors simultaneously measured the relaxation times of both qubits and studied the correlation between the obtained data for each of them.

In this article, we further cement the observation that the local TLS defects that are responsible for the fluctuations of the decoherence parameters of superconducting qubits are not significantly correlated and we show that these TLS defects will only affect each of the particular qubits within which they arise. In order to do so, we repeatedly estimate the relaxation times of the $5$-qubit quantum processors ibmq\_quito, ibmq\_belem, ibmq\_lima, ibmq\_santiago and ibmq\_bogota and then perform a correlation analysis on the measured fluctuations based on the Pearson correlation coefficient. We conclude that the obtained values of the Pearson correlation coefficients are not sufficiently high to observe significant correlation effects on the system. In consequence, we introduce the concept of fast time-varying quantum channels (FTVQC) as the appropriate mathematical model to describe the decoherence effects experienced by $n$-qubit superconducting systems. In this context, multi-qubit time varying quantum channels are then constructed with independent realizations of the decoherence parameters from qubit-to-qubit. In addition, we study the quantum channel capacity for the proposed family of quantum channels. Since the FTVQC channel resembles the classical scenario of fast fading, we discuss the ergodic quantum capacity, whose units are expressed in logical qubits per physical qubits, as a lower bound to the capacity of the FTVQC channels. Moreover, 
we prove that for the family of fast time-varying amplitude damping channels (FTVAD), the ergodic capacity is indeed equal to its quantum capacity. We also show that although a capacity loss is incurred in comparison to quantum channels that are assumed to be static, this change is not significant. Finally, we use numerical simulations to study the performance of planar codes and quantum turbo codes (QTC) when they operate over the FTVQC noise model. It is concluded that the performance of the codes worsens when compared to the static quantum channels, although this degradation is not as significant as the one codes experience over the previously considered multi-qubit TVQCs \cite{TVQC,outage}. Interestingly, we observe that the threshold of the surface codes deteriorates by a similar amount to the quantum capacity when the FTVQC multi-qubit model is considered.

\section{Time-varying quantum channels}

Time-varying quantum channels \cite{TVQC}, $\mathcal{N}(\rho,\omega,t)$, are defined as
\begin{equation}\label{eq:TVQCgen}
\mathcal{N}(\rho,\omega,t) = \sum_k E_k(\omega,t) \rho E_k^\dagger(\omega,t),
\end{equation}
where the $E_k(\omega,t)$ linear operators are the so-called Kraus operators of the operator-sum representation of a quantum channel, and are continuous-time random processes.

Decoherence arises from a wide range of physical processes involved in the interaction of qubits with their environment. In the context of supeconducting technologies, the principal vehicles for decoherence are energy relaxation and pure dephasing. The time-varying amplitude and phase damping quantum channel, $\mathcal{N}_\mathrm{APD}(\rho,\omega,t)$, \cite{TVQC} is a model that includes relaxation and pure dephasing effects, whose intensity (which is given by the relaxation time, $T_1$, and dephasing time, $T_2$) varies as a function of time. Note that whenever the Ramsey limit is saturated ($T_2\approx 2 T_1$), the channel is reduced to a time-varying amplitude damping channel, $\mathcal{N}_\mathrm{AD}(\rho,\omega,t)$ \cite{TVQC,outage,decoherenceBenchmarking}. The noise level of these quantum channels is characterized by the damping, $\{\gamma(\omega, t)\}$, and scattering, $\{\lambda(\omega,t)\}$, stochastic processes, which are functions of the qubit relaxation time $\{T_1(\omega,t)\}$ and the qubit dephasing time $\{T_2(\omega,t)\}$ as

\begin{equation}\label{eq:gammatime}
\gamma(\omega,t) = 1 - e^{-\frac{t}{T_1(\omega,t)}} \text{ and}
\end{equation}

\begin{equation}\label{eq:lambdatime}
\lambda(\omega,t) = 1 - e ^{\frac{t}{T_1(\omega,t)} - \frac{2t}{T_2(\omega,t)}}.
\end{equation}

The experimental analysis presented in \cite{decoherenceBenchmarking,TVQC} shows that $T_1(t,\omega)$ and $T_2(t,\omega)$ can be modelled by wide-sense stationary (WSS) random processes with means $\mu_{T_1},\mu_{T_2}$, standard deviations $\sigma_{T_1},\sigma_{T_2}$, and a stochastic coherence time, $T_\mathrm{c}$ in the order of minutes. Since the processing times for quantum algorithms and error correction rounds, $t_{\mathrm{algo}}$, are in the order of microseconds (the surface code cycle time is estimated to be $1 \mu s$ for superconducting devices) \cite{wallraff,TVQC,rsaRounds,googleSurfonemicro,oneMicrogidney}, $t_{\mathrm{algo}}\ll T_\mathrm{c}$, it is reasonable to assume that the processes remain constant during the execution of a quantum algorithm. In other words, $\{T_i(\omega,t)\}_{i=1}^2$ can be modelled as a set of random variables ($t=0$ has been selected without loss of generality due to the fact that the process is WSS.) $\{T_i(\omega)\}_{i=1}^2 = \{T_i(t, \omega)\rvert_{t=0}\}_{i=1}^2,\forall t\in[0,T],T<<T_\mathrm{c}$. Given that the random processes are assumed to be Gaussian, the random variables will also be Gaussian with distributions $\{\mathcal{N}(\mu_{T_i},\sigma^2_{T_i})\}_{i=1}^2$. However, since any and all realizations of $\{T_i(\omega)\}_{i=1}^2$ should be positive, they must be modelled as truncated Gaussian random variables in the region $[0,\infty]$. Therefore, the probability density functions are modelled as

\begin{equation}\label{eq:pdf}
f_{T_i}(t_i)=\begin{cases}
\frac{1}{\sigma_{T_i}\sqrt{2\pi}}\frac{\mathrm{e}^{-\frac{(t_i - \mu_{T_i})^2}{2\sigma_{T_i}^2}}}{1-\mathrm{Q}\left(\frac{\mu_{T_i}}{\sigma_{T_i}}\right)} &\text{ if } t_i\geq 0 \\
0 &\text{ if } t_i<0
\end{cases},
\end{equation}
where in the above expression, $i=\{1,2\}$ and $\mathrm{Q}(\cdot)$ is the Q-function defined as
\begin{equation}\label{josu6}
\mathrm{Q}(x) = \frac{1}{\sqrt{2\pi}}\int_x^{\infty}\mathrm{e}^{-\frac{x^2}{2}} dx.
\end{equation}

The twirled approximations of those time-varying quantum channels are also interesting since they can be simulated in an efficient manner using classical computers \cite{josuchannels,TVQC}. The time-varying Pauli twirl approximation (TVPTA), $\mathcal{N}_\mathrm{PTA}(\rho,\omega,t)$, is the Pauli channel \cite{josuchannels,TVQC} obtained by twirling a time-varying quantum channel by the $n$-fold Pauli group $\mathcal{P}_n$. Twirling the TVAD channel will lead to the Pauli channel (TVADPTA) described by the probabilities that each of the Pauli matrices has of taking place. Note that in this context these probabilities are realizations of the random processes:
\begin{equation}\label{eq:TVADPTA}
\begin{split}
& p_\mathrm{I}(\omega,t) = 1 - p_\mathrm{x}(\omega,t) - p_\mathrm{y}(\omega,t) - p_\mathrm{z}(\omega,t), \\
& p_\mathrm{x}(\omega,t) = p_\mathrm{y}(\omega,t) = \frac{1}{4}(1 - e^{-\frac{t}{T_1(\omega,t)}})\text{ and} \\
& p_\mathrm{z}(\omega,t) = \frac{1}{4}(1 + e^{-\frac{t}{T_1(\omega,t)}} - 2e^{-\frac{t}{2T_1(\omega,t)}}).
\end{split}
\end{equation}

For the TVAPD channel, the TVAPDPTA approximation is described by the realizations of the following stochastic processes for each of the Pauli matrices
\begin{equation}\label{eq:TVPTA}
\begin{split}
& p_\mathrm{I}(\omega,t) = 1 - p_\mathrm{x}(\omega,t) - p_\mathrm{y}(\omega,t) - p_\mathrm{z}(\omega,t), \\
& p_\mathrm{x}(\omega,t) = p_\mathrm{y}(\omega,t) = \frac{1}{4}(1 - e^{-\frac{t}{T_1(\omega,t)}})\text{ and} \\
& p_\mathrm{z}(\omega,t) = \frac{1}{4}(1 + e^{-\frac{t}{T_1(\omega,t)}} - 2e^{-\frac{t}{T_2(\omega,t)}}),
\end{split}
\end{equation}
where, once again, $T_1(\omega,t)$ and $T_2(\omega,t)$ are stochastic processes.

Another twirled channel of interest is the time-varying Clifford twirl approximation (TVCTA), $\mathcal{N}_\mathrm{CTA}(\rho,\omega,t)$ \cite{josuchannels,TVQC}, which for the TVAD channel will be a depolarizing channel with depolarizing parameter
\begin{equation}\label{eq:TVADCTA}
p(\omega,t) = \frac{3}{4} - \frac{1}{4}e^{-\frac{t}{T_1(\omega,t)}} - \frac{1}{2}e^{-\frac{t}{2T_1(\omega,t)}},
\end{equation}
and for the TVAPD channel a depolarizing channel with depolarizing parameter
\begin{equation}\label{eq:TVAPDCTA}
p(\omega,t) = \frac{3}{4} - \frac{1}{4}e^{-\frac{t}{T_1(\omega,t)}} - \frac{1}{2}e^{-\frac{t}{T_2(\omega,t)}},
\end{equation}
where, once more, $T_1(\omega,t)$ and $T_2(\omega,t)$ are stochastic processes.

\begin{table*}[]
\caption{Pearson correlation coefficients for the $T_1$ measurements obtained for the IBM quantum processors. Relaxation time is measured for two diferent calibration cycles for each of the machines. $r_{Q_iQ_j}$ is the obtained Pearson correlation coefficient for the relaxation times of qubit $i$ and qubit $j$ of the processor. Bootstrapping was used to determine the $95\%$ confdence intervals, presented in parentheses.}
\begin{tabular}{|l|l|l|l|l|l|l|l|}
\hline
\textbf{Scenario} & $r_{Q_0Q_1} $ & $r_{Q_0Q_2} $ & $r_{Q_0Q_3} $ & $r_{Q_0Q_4} $ & $r_{Q_1Q_2} $   \\ \hline
ibmq\_quito S1  &      $-0.207(-0.066,-0.35)$          &          $0.036(-0.094,0.17)$           &     $-0.287(-0.17,-0.41)$         &           $-0.076(-0.194,0.049)$         &      $-0.083(-0.223,0.05)$                      \\ \hline
ibmq\_quito S2  &       $0.096(0.022,0.167)$        &            $0.154(0.054,0.256)$         &      $0.53(0.444,0.6)$       &       $0.133(0.052,0.26)$  &     $0.024(-0.081,0.125)$        \\ \hline
 ibmq\_belem S1   &        $-0.002(-0.081,0.173)$       &           $0.013(-0.248,0.024)$           &       $-0.008(-0.047,0.059)$       &      $-0.019(-0.033,0.163)$                &              $-0.028(-0.085,0.03)$                   \\ \hline
ibmq\_belem S2    &       $-0.141(-0.23,-0.05)$       &         $-0.153(-0.25,-0.057)$            &      $-0.15(-0.24,-0.06)$        &         $-0.38(-0.46,-0.29)$           &     $0.2(0.096,0.3)$         \\ \hline
  ibmq\_lima S1  &        $-0.028(-0.09,0.03)$       &           $-0.09(-0.16,-0.024)$         &      $0.295(0.24,0.35)$        &          $0.27(0.221,0.31)$ &       $-0.024(-0.086,0.037)$       \\ \hline
  ibmq\_lima S2  &      $0.244(0.131,0.346)$         &         $0.347(0.24,0.45)$             &       $-0.0049(-0.125,0.113)$       &         $0.044(-0.077,0.166)$           &      $-0.011(-0.138,0.13)$                  \\ \hline
  ibmq\_santiago S1  &       $-0.035(-0.175,0.11)$        &         $0.1(-0.033,0.255)$            &      $0.16(0.005,0.31)$       &         $-0.05(-0.17,0.06)$ &            $-0.035(-0.18,0.11)$              \\ \hline
  ibmq\_bogota S1  &       $0.15(0.006,0.29)$        &        $0.093(-0.05,0.22)$              &      $-0.066(-0.19,0.055)$        &         $-0.11(-0.23,0.014)$          &            $-0.032(-0.17,0.11)$  \\ \hline
\end{tabular}
\centering
\begin{tabular}{|l|l|l|l|l|l|}
\hline
\textbf{Scenario}& $r_{Q_1Q_3} $ & $r_{Q_1Q_4} $  &  $r_{Q_2Q_3} $  & $r_{Q_2Q_4} $  & $r_{Q_3Q_4} $  \\ \hline
ibmq\_quito S1  &            $0.116(-0.032,0.262)$  &         $0.0031(-0.128,0.138)$     &          $-0.128(-0.25,0.002)$            &     $-0.161(-0.275,-0.025)$        &   $-0.0159(-0.154,0.133)$   \\ \hline
ibmq\_quito S2  & $0.0028(-0.078,0.08)$ &      $-0.093(-0.197,0.003)$       &         $0.031(-0.051,0.115)$            &      $0.049(-0.0438,0.14)$      &    $0.0799(-0.00441,0.159)$  \\ \hline
 ibmq\_belem S1   &  $0.059(10^{-5},0.117)$ &      $0.297(0.233,0.35)$        &         $0.069(-0.001,0.142)$           &    $0.023(-0.031,0.079)$       &      $-0.004(-0.055,0.04)$      \\ \hline
ibmq\_belem S2    & $0.117(0.026,0.205)$ &    $0.227(0.131,0.317)$         &          $0.062(-0.03,0.152)$            &      $-0.027(-0.131,0.069)$      &      $0.214(0.115,0.3)$      \\ \hline
  ibmq\_lima S1  &  $0.057(-0.006,0.11)$  &     $-0.033(-0.09,0.025)$       &         $0.163(0.1,0.23)$             &      $0.196(0.13,0.26)$     &    $0.49(0.437,0.54)$      \\ \hline
  ibmq\_lima S2  & $-0.013(-0.156,0.126)$  &      $-0.085(-0.2,0.03)$      &          $0.047(-0.085,0.18)$            &      $0.147(0.04,0.25)$       &   $-0.13(-0.25,-0.008)$  \\ \hline
  ibmq\_santiago S1  & $-0.66(-0.21,0.065)$ &     $-0.27(-0.39,-0.14)$         &       $0.162(0.017,0.3)$         &   $0.063(-0.11,0.21)$    &     $0.078(-0.07,0.22)$        \\ \hline
  ibmq\_bogota S1  & $0.19(0.054,0.33)$  &   $0.267(0.13,0.39))$       &            $-0.09(-0.22,0.051)$          &    $-0.086(-0.21,0.05)$       &      $0.13(-0.02,0.27)$      \\ \hline
\end{tabular}
\label{tab:T1corr}
\end{table*}

\section{Multi-qubit time-varying quantum channels}
Time-varying quantum channels (TVQC) describe the coherence loss of a qubit when the relaxation and dephasing times that describe the rate of interaction between the qubit and its environment fluctuate as functions of time for the same cooldown \cite{TVQC}. The proposal of this theoretical framework of quantum noise was motivated by the repeated observation of such intra-cooldown stochastic behaviour of superconducting qubit decoherence parameters in the literature \cite{decoherenceBenchmarking,klimov,fluctAPS,fluctApp,SchlorPhD,fluctErrorBars,fluctDeph}.

TVQCs successfully account for the experimentally observed time-varying nature of the decoherence experienced by single superconducting qubits. However, quantum information processing tasks (algorithms, error correction, memories or communications) require sets of qubits to appropriately achieve the tasks that they are designed for. Thus, it is necessary to consider multi-qubit time-varying quantum channels to accurately assess the impact of $T_1$ and $T_2$ fluctuation on practical quantum computing. In this section, we discuss the way in which the noise of such multi-qubit systems can be modelled when time-fluctuations are present. In order to do so, we study the locality of the time flcutuations before discussing the multi-qubit time-varying noise models for superconducting chips.

\subsection{Decoherence parameter fluctuations are local to each qubit}\label{res:exp}
We have conducted simultaneous measurements of the relaxation times of the IBM quantum processors ibmq\_quito, ibmq\_belem, ibmq\_lima, ibmq\_santiago and ibmq\_bogota \cite{IBMqexp} spaced out over time (See Appendix \ref{appC} for a detailed description of the experiments). All of these quantum processors are comprised of $5$ superconducting qubits (with different architectures and connectivity). The aim of this experiment is to verify that the fluctuations of the decoherence parameters of each of the constituent qubits of these $5$-qubit superconducting processors are local to the particular qubits themselves. In this way, we want to extend the analysis undertaken in \cite{decoherenceBenchmarking} for a $2$-qubit processor to more complex Noisy Intermediate-Scale Quantum (NISQ) devices.

Table \ref{tab:T1corr} shows the results obtained for the Pearson correlation tests we conducted on the measured relaxation times (See Appendix \ref{appA} for the description of the statistical analysis). Based on these outcomes, it is clear that the $T_1$ fluctuations are not significantly correlated between the qubits of the systems in the majority of the scenarios. In fact, the values of the correlation coefficients are all below the threshold of significant correlation (which stands at $0.6$) that is generally considered for classical scenarios (See Appendix \ref{appA}). This holds even when considering $95\%$ confidence intervals. These results support the hypothesis that the fluctuations of the decoherene parameters of superconducting qubits are caused by local effects, making it reasonable to assume that the local TLS defects responsible for these fluctuations are not coupled amongst themselves. We must note that there is an outlier among our results where correlation appears to be present. This occurs for qubits $0$ and $3$ of the ibmq\_quito S2 scenario, whose correlation coefficient has a value of $0.53$ and the upper limit of the confidence interval is set to $0.6$. However, the $0^{th}$ qubit of such scenario shows a step size transition at the end of the experiment (See Appendix \ref{appC}), an effect that may have impacted the value of the obtained Pearson correlation. In fact, considering just the samples before the hard transition, the results are $r_{Q_0Q_3}=0.091(-0.026,0.217)$, which are well within the range of uncorrelated values. In light of these results, it is safe to assume that the fluctuations of the decoherence parameters of superconducting qubits will be due to local effects.

\subsection{Fast time-varying quantum channels}\label{res:FTVQC}
In \cite{TVQC,outage}, the authors proposed a multi-qubit time-varying quantum channel under the assumption that the relaxation and dephasing times are constant for all the qubits inside a particular error correction block, but vary from block to block. Thus, a realization of the multi-qubit time-varying quantum channel for block $m\in\mathbb{N}$ of duration $t_\mathrm{algo}$ can be described mathematically as

\begin{equation}\label{eq:slowNTVQC}
\begin{split}
&\mathcal{N}^{(n)}(\rho,t_1^m,t_2^m,t=t_\mathrm{algo}) = \mathcal{N}^{\otimes n}(\rho,t_1^m,t_2^m,t=t_\mathrm{algo})\\  &= \sum_{\mathcal{E}\in(\{E_k\}_k)^{\otimes n}} \mathcal{E}(t_1^m,t_2^m,t=t_\mathrm{algo})\rho \mathcal{E}^\dagger(t_1^m,t_2^m,t=t_\mathrm{algo}),
\end{split}
\end{equation}
where $t_1^m$ and $t_2^m$ refer to the realizations of the sequences of independent random variables $\{T_1^m(\omega)\}_{m\in \mathbb{N}}$ and $\{T_2^m\}(\omega)\}_{j\in \mathbb{N}}$, respectively, and $\mathcal{E}(t_1^m,t_2^m,t=t_\mathrm{algo}) = \mathcal{E}_1(t_1^m,t_2^m,t=t_\mathrm{algo})\otimes\cdots\otimes \mathcal{E}_{n-1}(t_1^m,t_2^m,t=t_\mathrm{algo})\otimes\mathcal{E}_n(t_1^m,t_2^m,t=t_\mathrm{algo})$ with $\mathcal{E}_j(t_1^m,t_2^m,t=t_\mathrm{algo})\in\{E_k(t_1^m,t_2^m,t=t_\mathrm{algo})\}_{k}$ referring to the Kraus operators of the single qubit TVQCs associated to those realizations of the decoherence parameters. $(\{E_k\}_k)^{\otimes n}$ refers to the set of $n$-fold tensor products of the Kraus operators of the single-qubit TVQCs. Note that, here, those Kraus operators are related to some TVQC that depends on the relaxation and dephasing times. This, however, does not exclude the construction of similar time-varying quantum channels depending on other parameters that may show similar behaviour to $T_1$ and $T_2$.

This means that the multi-qubit channel considered in \cite{TVQC,outage} assumes that the realizations of the random variables that describe the noise experienced by each of the qubits of the system are identical, i.e, that these random variables are perfectly correlated. As discussed in \cite{TVQC,outage}, this model is reminiscent of the classical slow fading scenario. For simplicity, we will adopt this terminology to refer to the multi-qubit TVQCs that add perfectly correlated noise to each qubit. Hence, we name these types of channels as slow time-varying quantum channels (STVQC). If we recall our discussion in the previous section, we now know that the STVQCs considered in \cite{TVQC,outage} are not the most accurate type of multi-qubit time-varying quantum channels for the superconducting NISQ devices we are considering in this paper. This has to do with the fact that the origin of the decoherence parameter fluctuations of superconducting qubits are local to each qubit.

Based on our discussions thus far, we know that the way to construct multi-qubit time-varying quantum channels to accurately model the superconducting hardware considered in this article is by considering that the individual TVQCs that make up the multi-qubit channel for each QEC block, $m\in\mathbb{N}$, are defined by sequences of random variables $\{T_1^{m^j}(\omega)\}_{j=1}^n$ and $\{T_2^{m^j}(\omega)\}_{j=1}^n$ whose elements are independent amongst themselves. Thus, the realizations of the decoherence parameter random variables will not only be independent from block to block, but also from qubit to qubit inside a block. In this way, a realization of the multi-qubit time-varying quantum channel for block $m\in\mathbb{N}$ of duration $t_\mathrm{algo}$ can be described mathematically as
\begin{equation}\label{eq:fastNTVQC}
\begin{split}
&\mathcal{N}^{(n)}(\rho,\{t_1^{m^j}\}_{j=1}^n,\{t_2^{m^j}\}_{j=1}^n,t=t_\mathrm{algo})  \\& = \bigotimes_{j=1}^n \mathcal{N}(\rho,t_1^{m^j},t_2^{m^j},t=t_\mathrm{algo})\\  &= \!\!\sum_{\mathcal{E}\in\bigotimes_{j=1}^n(\{E_{k}(t_1^{m^j},t_2^{m^j})\}_{k})} \!\!\!\!\!\!\!\!\!\!\!\!\!\!  \mathcal{E}(\{t_1^{m^j}\}_{j=1}^n,\{t_2^{m^j}\}_{j=1}^n)\rho \mathcal{E}^\dagger(\{t_1^{m^j}\}_{j=1}^n,\{t_2^{m^j}\}_{j=1}^n),
\end{split}
\end{equation}
where $\{t_1^{m^j}\}_{j=1}^n$ and $\{t_2^{m^j}\}_{j=1}^n$ refer to the realizations of the sequences of independent random variables $\{T_1^{m^j}(\omega)\}_{j=1}^n$ and $\{T_2^{m^j}(\omega)\}_{j=1}^n$, respectively, and $\mathcal{E}(\{t_1^{m^j}\}_{j=1}^n,\{t_2^{m^j}\}_{j=1}^n) = \mathcal{E}_1(t_1^{m^1},t_2^{m^1})\otimes\cdots\otimes \mathcal{E}_{n-1}(t_1^{m^{n-1}},t_2^{m^{n-1}})\otimes\mathcal{E}_n(t_1^{m^n},t_2^{m^n})$ with $\mathcal{E}_j(t_1^{m^j},t_2^{m^j})\in\{E_k(t_1^{m^j},t_2^{m^j})\}_{k}$ referring to the Kraus operators of the single qubit TVQCs associated to those realizations of the decoherence parameters, where we have incurred in the slight abuse of notation $E_k(\cdot,\cdot) = E_k(\cdot,\cdot,t=t_{algo})$. Note that the sequence of random variables will also be independent from block to block, i.e., the elements of $\{T_1^{m^j}(\omega)\}_{m,j}, \forall m,j\in\mathbb{N}$ and $\{T_2^{m^j}(\omega)\}_{m,j}, \forall m,j\in\mathbb{N}$ will be independently distributed, respectively.

In this way, each of the realizations of this multi-qubit TVQC will have a different noise ``intensity'' (the actual noise operators will be the same, but the noise level will change) for each of the qubits of the superconducting quantum processor. This model resembles the classical scenario of fast fading \cite{fading,tse}. In such scenarios, the fading process changes so quickly that each of the symbols of a transmitted codeword is subjected to a different fading gain (where the fading gain for each of the symbols is independent) and, thus, to different noise levels. Note that for the multi-qubit TVQC that we are discussing, the values of the $T_1$ and $T_2$ parameters will vary slowly, but since their particular realizations are independent from qubit to qubit, the channel actually resembles to the fast fading scenario. Thus, we will once again borrow from the classical realm and refer to these quantum channels as fast time-varying quantum channels (FTVQC).

In this context, we will refer to the most commonly considered construction of multi-qubit channels as static quantum channels. A widespread assumption in the QECC community is that all the qubits of a quantum processor experience the same noise through time \cite{josuchannels,TVQC}. This implies that all of the qubits of the system have the same decoherence parameters, $T_1$ and $T_2$, and that these will not vary with the passage of time. Thus, multi-qubit static channels are constructed by evaluating the Kraus operators of the channels with the mean values of those parameters. Consequently, the static multi-qubit time-varying quantum channel for every block $\forall m\in\mathbb{N}$ of duration $t_\mathrm{algo}$ is described mathematically as

\begin{equation}\label{eq:staticVQC}
\begin{split}
&\mathcal{N}^{(n)}(\rho,\mu_{T_1},\mu_{T_2},t=t_\mathrm{algo}) = \mathcal{N}^{\otimes n}(\rho, \mu_{T_1},\mu_{T_2},t=t_\mathrm{algo})\\  &= \sum_{\mathcal{E}\in(\{E_k\}_k)^{\otimes n}} \mathcal{E}(\mu_{t_1},\mu_{t_2},t=t_\mathrm{algo})\rho \,\mathcal{E}^\dagger(\mu_{t_1},\mu_{t_2},t=t_\mathrm{algo}),
\end{split}
\end{equation}
where it can be seen that the channel will be equal for all the QEC blocks, since $\mathcal{E}(\mu_{T_1},\mu_{T_2}) = \mathcal{E}_1(\mu_{T_1},\mu_{T_2},t=t_\mathrm{algo})\otimes\cdots\otimes \mathcal{E}_{n-1}(\mu_{T_1},\mu_{T_2},t=t_\mathrm{algo})\otimes\mathcal{E}_n(\mu_{T_1},\mu_{T_2},t=t_\mathrm{algo})$ with $\mathcal{E}_j(\mu_{T_1},\mu_{T_2},t=t_\mathrm{algo})\in\{E_k(\mu_{T_1},\mu_{T_2},t=t_\mathrm{algo})\}_{k}$ is independent of the block $m$. $(\{E_k\}_k)^{\otimes n}$  denotes the set of $n$-fold tensor products of the Kraus operators of the single-qubit TVQCs.

\section{Quantum channel capacity}
The quantum channel capacity, $C_Q$, for a static quantum channel $\mathcal{N}$, is defined as the supremum of all achievable quantum coding rates (the quantum coding rate is defined as $R_\mathrm{Q} = k/n$, where $k$ is the number of logical qubits and $n$ is the number of physical qubits). A rate, $R_Q$, is said to be achievable for $\mathcal{N}$ if there exists a sequence of $[[n,k]]$ quantum codes of rate $R_Q$ such that the probability of error of the codes goes to zero as the blocklength $n$ of the code goes to infinity, $n\rightarrow \infty$.  The definition of the quantum capacity, often referred to as the Lloyd-Shor-Devetak (LSD) capacity, is given by the following theorem \cite{wildeQIT,quantumcap}.
\begin{thm}[LSD capacity]\label{thm:quantumcap}
\textit{The quantum capacity, $C_\mathrm{Q}(\mathcal{N})$, of a quantum channel, $\mathcal{N}$, is equal to the regularized coherent information of the channel
\begin{equation}\label{eq:cqreg}
C_\mathrm{Q}(\mathcal{N}) = Q_{\mathrm{reg}}(\mathcal{N}),
\end{equation}
where
\begin{equation}\label{eq:coherentreg}
Q_{\mathrm{reg}}(\mathcal{N}) = \lim_{n\rightarrow \infty} \frac{1}{n} Q_{\mathrm{coh}}(\mathcal{N}^{\otimes n}).
\end{equation}
The channel coherent information, $Q_{\mathrm{coh}}(\mathcal{N})$, is defined as
\begin{equation}\label{eq:coh}
Q_{\mathrm{coh}}(\mathcal{N}) = \max_{\rho} (S(\mathcal{N}(\rho)) - S(\rho_\mathrm{E})),
\end{equation}
where $S$ is the von Neumann entropy and $S(\rho_\mathrm{E})$ measures how much information the environment has.}
\end{thm}

There is no general single-key formula to compute the regularization necessary to calculate the quantum channel capacity given in theorem \ref{thm:quantumcap}. This is due to the fact that the coherent information of the channel is not generally additive \cite{wildeQIT,bothAdd}. However, for specific classes of degradable quantum channels, such as the amplitude damping (AD) channel, the channel coherent information has been shown to be additive, reducing the expression of the regularization shown earlier to a single-letter formula:
\begin{equation}
C_Q(\mathcal{N})= \lim_{n\rightarrow \infty}\frac{1}{n} Q_\mathrm{coh}(\mathcal{N}^{\otimes n}) = \lim_{n\rightarrow \infty}\frac{1}{n} n Q_\mathrm{coh}(\mathcal{N})=Q_\mathrm{coh}(\mathcal{N}),
\end{equation}
that is, the quantum channel capacity is actually the same as the channel coherent information. Degradable and anti-degradable quantum channels are defined as \cite{wildeQIT}

\begin{definition}[Degradable and anti-degradable channels]
A channel $\mathcal{N}$ from system $A$ to system $B$ is said to be degradable if there exists a CPTP map $\mathcal{D}$ from system $B$ to the environment $E$ such that $\mathcal{N}^c = \mathcal{D}\circ \mathcal{N}$, where $\mathcal{N}^c$ is named the complementary channel from system $A$ to the environment $E$. Additionally, a channel $\mathcal{N}$ from system $A$ to system $B$ is said to be anti-degradable if there exists a CPTP map $\mathcal{D}$ from the environment $E$ to system $B$ such that $\mathcal{N} = \mathcal{D}\circ \mathcal{N}^c$, where $\mathcal{N}^c$ is named the complementary channel from system $A$ to the environment $E$.
\end{definition}

Both degradable and anti-degradable channels have the property of having additive channel coherent information \cite{bothAdd}. Moreover, anti-degradable channels have always vanishing quantum channel capacity \cite{wildeQIT}.

The quantum capacity of an AD channel with damping parameter $\gamma \in [0,1]$ is equal to \cite{wildeQIT,quantumcap}
\begin{equation}\label{eq:ADcap}
C_\mathrm{Q}(\gamma) = \max_{\xi\in[0,1]}  \mathrm{H}_2((1-\gamma) \xi) - \mathrm{H}_2(\gamma \xi),
\end{equation}
whenever $\gamma \in [0, 1/2]$, and zero for $\gamma\in [1/2,1]$. $\mathrm{H}_2(x)$ is the binary entropy. This comes from the fact that the AD channel is a degradable channel for $\gamma\in [0,1/2]$ and anti-degradable for $\gamma\in [1/2,1]$.

An expression for the quantum capacity of the widely used Pauli channels remains unknown since its coherent information is not additive \cite{josuchannels,wildeQIT}. However, a lower bound that can be achieved by stabilizer codes, the hashing bound (which equals the single-qubit coherent information of the channel), $C_\mathrm{H}$, \cite{wildeQIT} is known. The reason why the quantum capacity of a Pauli channel can be higher than the hashing bound (this is the same as saying that the coherent information is superadditive), i.e. $C_\mathrm{Q} \geq C_\mathrm{H}$, is the degenerate nature of quantum codes \cite{degenPRL,degen, degen_2}, a quantum-exclusive phenomenon through which several distinct channel errors affect encoded quantum states in an indistinguishable manner. In fact, the depolarizing channel has been proven to be superadditive for vey noisy depolarizing probabilities \cite{wildeQIT}.

The hashing bound for a Pauli channel defined by the probability mass function $\mathbf{p} = (p_\mathrm{I},p_\mathrm{x},p_\mathrm{y},p_\mathrm{z})$ is given by \cite{wildeQIT}
\begin{equation}\label{eq:hash}
C_\mathrm{H}(\mathbf{p}) = 1 - \mathrm{H}_2(\mathbf{p}).
\end{equation}
$\mathrm{H}_2(\mathbf{p}) = -\sum_j p_j\log_2(p_j)$ is the entropy in bits of a discrete random variable with probability mass function given by $\mathbf{p}$.

\subsection{Classical fast fading channels}\label{meth:fastfading}
As stated in the previous section, the FTVQC model is similar to the classical scenario of fast fading channels and, thus, we will introduce the capacity of the latter channels before we discuss the one of their quantum couterparts. Consider the classical scenario where the received symbols, $y[m]$, are given by
\begin{equation}
y[m] = \alpha[m]x[m] + w[m],
\end{equation}
where $x[m]$ refers to the transmitted symbols, $\alpha[m]$ refers to the fading gains, and $w[m]$ refers to i.i.d. complex additive white Gaussian noise \cite{tse}. If the fading process has a stochastic coherence time that is lower than the duration of a symbol, then the set of fading gains will be given by the realization of a sequence of i.i.d. random variables. This classical scenario is known as the fast fading channel, and has a well defined capacity, known as the ergodic capacity, defined as \cite{tse}
\begin{equation}\label{eq:ergo}
C_\mathrm{erg} = \mathrm{E}\{C(\omega)\} = \mathrm{E}\{\log(1+|\alpha(\omega)|^2\mathrm{SNR})\},
\end{equation}
where $\mathrm{SNR}$ refers to the signal-to-noise ratio. The intuition behind this limit is that one can average over many independent fades of the channel by coding over a large number of coherence time intervals. In this way, a reliable rate given by the mean of the ``instantaneous'' capacities can indeed be achieved.

\subsection{Capacity of FTVQCs}\label{sub:ergcap}
We have seen that the quantum capacity is the maximum rate at which quantum information can be communicated/corrected over many independent uses of a noisy quantum channel. Therefore, the quantum channel capacity is a quantity of interest for quantum coding theorists, as it represents the maximum rate at which QECCs can correct the effects of a specific noise map. 

For this reason, studying the quantum capacity of time-varying quantum channels becomes a fundamental task if we are to correctly design QECCs to operate over such decoherence models. In  \cite{outage}, the concept of the quantum outage probability as the asymptotically achievable error rate was discussed for slow time-varying quantum channels. The need for a new concept like the outage probability stems from the fact that the quantum capacity of the aforementioned channels is strictly zero, which makes it necessary to employ other information theoretical limits to study this particular family of channels. Given the similarity that exists between slow time-varying quantum channels and the classical scenario of slow fading, it makes good sense to adapt the outage probability (the theoretical quantity used to study classical slow fading) to the quantum paradigm. This has resulted in the proposition of the quantum outage probability as the most appropriate metric to study STVQCs \cite{outage}.

Let us now consider fast time-varying quantum channels (FTVQC). Because the capacity of these channels is not strictly zero (generally), we will be able to use it to determine the maximum coding rate over these channels. As stated in the previous section, the elements of the sequences of random variables $\{T_1^j(\omega)\}_{j=1}^n$ and $\{T_2^j(\omega)\}_{j=1}^n$ in a particular QEC block are independent. From this point we will assume that those elements are identically distributed too. The rationale behind this is that one needs to know how the $T_1$ and $T_2$ of each of the qubits is distributed (in general each qubit might have different means and standard deviations) in order to analyze the capacity and to perform numerical simulations. Under such assumption, a lower bound for the the quantum capacity of the combined amplitude and phase damping FTVQC is provided in the following theorem.

\begin{thm}[Quantum capacity of APD FTVQCs]
\textit{The quantum capacity of the combined amplitude and phase damping FTVQC is lower bounded by its ergodic quantum channel capacity
\begin{equation}\label{eq:CergAPD}
\begin{split}
C_\mathrm{Q}(\bar{\gamma},\bar{\lambda}) &  \geq C_\mathrm{Q}^\mathrm{erg}(\bar{\gamma},\bar{\lambda}) \\& = \mathrm{E}\{Q_\mathrm{coh}(\omega)\} = \int\int Q_\mathrm{coh}(\gamma,\lambda) p_{\gamma,\lambda}(\gamma,\lambda) d\gamma d\lambda \\
&= \int \int Q_\mathrm{coh}(t_\mathrm{algo},t_1,t_2)p_{T_1,T_2}(t_1,t_2) dt_1 dt_2, 
\end{split}
\end{equation}
where $Q_\mathrm{coh}$ refers to the channel coherent information, operator $\mathrm{E}\{\cdot\}$ is the mean, and $\bar{\gamma},\bar{\lambda}$ refer to the damping and scattering probabilities defined by the mean relaxation and dephasing times.}
\end{thm}
\begin{proof}
We can actually use the quantum channel capacity to quantify the maximum coding rate that can be achieved over fast time varying quantum channels. To clarify this, let us look at how the capacity of a FTVQC is computed. We first obtain the realizations $\{t_1^{m^j}\}_{j=1}^n$ and $\{t_2^{m^j}\}_{j=1}^n$ for a block $m\in\mathbb{N}$ of the sequences of the relaxation $\{T_1^{m^j}(\omega)\}_{j=1}^n$ and dephasing $\{T_2^{m^j}(\omega)\}_{j=1}^n$ random variables and integrate them in the FTVQC channel model (See equation \eqref{eq:fastNTVQC}). Then, we can bound the capacity of the realization of the fast time-varying quantum channel (which will be fixed) for block $m\in\mathbb{N}$ as
\begin{equation}\label{k3}
\begin{split}
&C_\mathrm{Q}\left(\mathcal{N}(\rho,\{t_1^{m^j}\}_{j=1}^n,\{t_2^{m^j}\}_{j=1}^n)\right) \\ & = \lim_{n\rightarrow \infty}\frac{1}{n}Q_\mathrm{coh}\left(\bigotimes_{j=1}^n \mathcal{N}(\rho,t_1^{m^j},t_2^{m^j})\right)  \\ & \geq \lim_{n\rightarrow \infty}\frac{1}{n}\sum_{k=1}^n Q_\mathrm{coh}\left(\mathcal{N}(\rho,t_1^{m^j},t_2^{m^j})\right),
\end{split}
\end{equation}
where we incurred in the abuse of notation $\mathcal{N}(\cdot,\cdot,\cdot) = \mathcal{N}(\cdot,\cdot,\cdot, t= t_\mathrm{algo})$ for simplicity. The inequality arises from the fact that the channel coherent information might be superadditive in general, i.e. $Q_\mathrm{coh}(\mathcal{N}\otimes\mathcal{M})\geq Q_\mathrm{coh}(\mathcal{N})+Q_\mathrm{coh}(\mathcal{M})$. The combined amplitude and phase damping channel has been proven not to be degradable for the region where its coherent information is positive \cite{APDcap} and, thus, the additivity of its coherent information remains an open question. For the Pauli channels (the twirled approximations of the are included in this family), the channel coherent information has been proven to be strictly superadditive for some channel instances (very noisy depolarizing channel, for example). 

Moreover, note that the sequence of values $\{t_1^{m^j}\}_{j=1}^n$ and $\{t_2^{m^j}\}_{j=1}^n$ specify the relaxation and dephasing parameters for each one of the $n$ qubits in the $m^{\mathrm{th}}$ block. Therefore, the channel capacity in \eqref{k3} will depend on such realizations of the decoherence parameters, which means that, once again, the channel capacity will itself become a random variable, $C_Q(\omega)$. In fact, the bound we derived will also become a random variable, $\lim_{n\rightarrow \infty}\frac{1}{n}\sum_{k=1}^n Q_\mathrm{coh}\left(\mathcal{N}(\rho,T_1^{m^j}(\omega),T_2^{m^j}(\omega))\right)$.
However, because of the law of large numbers and due to the fact that the elements of the sequences of random variables $\{T_1^{m^j}(\omega)\}_{j=1}^n$ and $\{T_2^{m^j}(\omega)\}_{j=1}^n$ are independent (this is the case for FTVQCs) and identically distributed (assumed before) the following holds
\begin{equation}\label{eq:proofConv}
\begin{split}
&\lim_{n\rightarrow \infty}\frac{1}{n}\sum_{k=1}^n Q_\mathrm{coh}\left(\mathcal{N}(\rho,t_1^{m^j},t_2^{m^j})\right)  \\&  = \mathrm{E}\left\lbrace Q_\mathrm{coh}\left(\mathcal{N}(\rho,T_1^{m^j}(\omega),T_2^{m^j}(\omega)\right)\right\rbrace = \mathrm{E}\{Q_\mathrm{coh}(\omega)\}.
\end{split}
\end{equation}
The above expression is true for almost all realizations of the sequences of random variables of the decoherence parameters, or similarly, for almost all blocks $m$. This is similar to what happens with the channel capacity of classical fast fading channels \cite{fading,tse}, where the capacity is usually referred to as the ergodic channel capacity. That is why we will refer to the quantity in \eqref{eq:proofConv} as the ergodic quantum channel capacity, $C_\mathrm{Q}^\mathrm{erg}$, and the bound in \eqref{k3} can be written as
\begin{equation}\label{eq:capboundfinal}
C_\mathrm{Q}(\mathcal{N}) \geq C_\mathrm{Q}^\mathrm{erg}(\mathcal{N}) = \mathrm{E}\{Q_\mathrm{coh}(\omega)\}.
\end{equation}
\end{proof}

 In the obtained bound, th ergodic quantum capacity is a function of $T_1$ and $T_2$, as well as of $t_\mathrm{algo}$. This comes from the fact that $\gamma$ and $\lambda$ are functions of those three parameters, and in the integral $t_\mathrm{algo}$ is fixed (its value can be obtained from $\bar{\gamma},\bar{\lambda}$). Note that the quantum capacity of the combined amplitude and phase damping channel is unknown at the time of writing \cite{APDcap}. This comes from the fact that the channel is not degradable in the region where the channel coherent information is positive, implying that it might be superadditive \cite{APDcap}. This unknown superadditivity problem extends to its FTVQC version, implying that we can only provide a bound.

\begin{figure}[!ht]
\centering
\includegraphics[width=\linewidth]{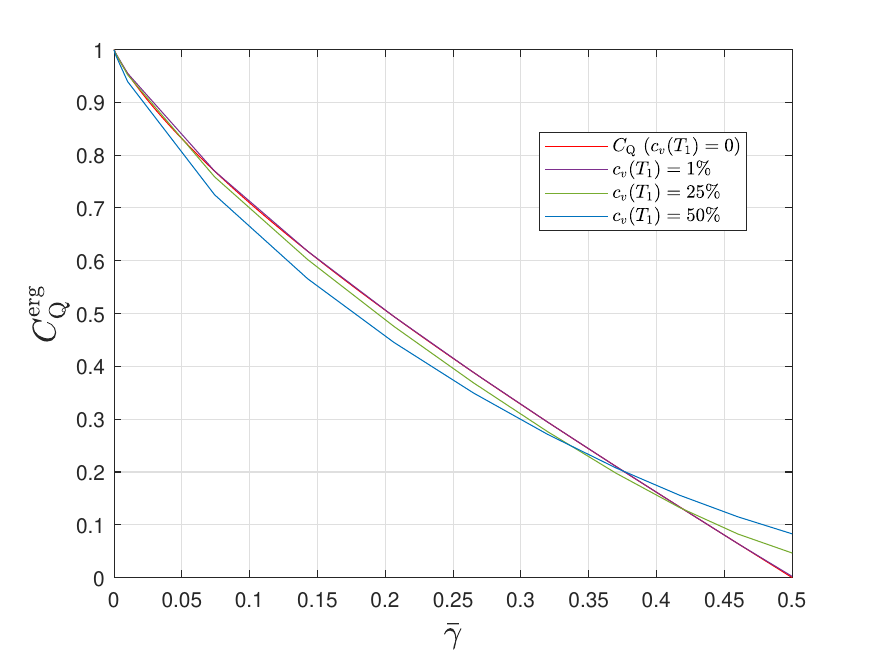}
\caption{\textbf{Quantum capacity of the FTVAD channel.} The metric is calculated for FTVADs with $c_\mathrm{v}(T_1)=\{1,25,50\}\%$.}
\label{fig:CergAD}
\end{figure}

\begin{figure*}[!ht]
\centering
\includegraphics[width=\linewidth]{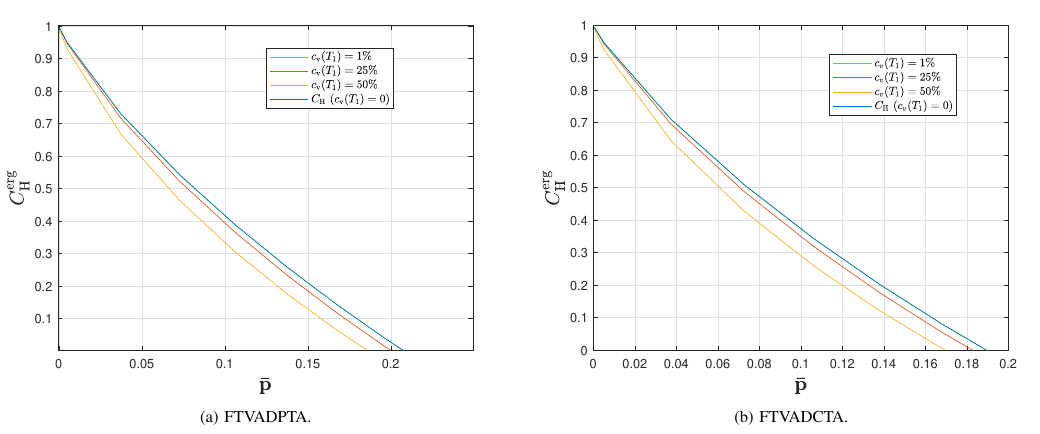}
\caption{\textbf{Ergodic hashing bounds for the twirled approximations of the AD channel.} The metric is calculated for $c_\mathrm{v}(T_1)=\{1,25,50\}\%$. \textbf{a} fast time-varying amplitude damping Pauli twirl approximation. \textbf{b} fast time-varying amplitude damping Clifford twirl approximation.}
\label{fig:CerghashAD}
\end{figure*}

The time-varying amplitude damping channel (TVAD) is typically used to describe the noise that manifests when working with $T_1$-limited qubits, such as those of \cite{decoherenceBenchmarking}, since the pure dephasing effects suffered by this type of qubits are negligible (the Ramsey limit, $T_2\approx 2T_1$ is saturated). The following theorem shows how the quantum channel capacity coincides with the ergodic capacity for this class of FTVQCs.
\begin{thm}[Quantum capacity of amplitude damping FTVQCs]
\textit{The quantum capacity of the amplitude damping FTVQC is equal to its ergodic quantum channel capacity
\begin{equation}\label{eq:CergAD}
\begin{split}
C_\mathrm{Q}(\bar{\gamma})=C_\mathrm{Q}^\mathrm{erg}(\bar{\gamma}) &= \mathrm{E}\{Q_\mathrm{coh}(\omega)\} \\&= \int Q_\mathrm{coh}(\gamma) p_{\gamma}(\gamma) d\gamma \\
&= \int Q_\mathrm{coh}(t_\mathrm{algo},t_1)p_{T_1}(t_1) dt_1,
\end{split}
\end{equation}
where $Q_\mathrm{coh}$ refers to the channel coherent information, operator $\mathrm{E}\{\cdot\}$ is the mean, and $\bar{\gamma}$ refers to the damping probabilities defined by the mean relaxation time.}
\end{thm}
\begin{proof}
We need to prove that the bound $C_\mathrm{Q}(\mathcal{N}) \geq C_\mathrm{Q}^\mathrm{erg}(\mathcal{N})$ is actually an equality. To do so, we must prove that the channel coherent information in \eqref{k3} is actually additive for the family of FTVAD channels. The static amplitude damping channel exhibits degradability for damping parameters $\gamma\in [0,1/2]$ and antidegradability otherwise, i.e. $\gamma\in [1/2,1]$. In this way, the tensor product of the limit
\begin{equation}
\lim_{n\rightarrow \infty}\frac{1}{n}Q_\mathrm{coh}\left(\bigotimes_{j=1}^n \mathcal{N}(\rho,t_1^{m^j},t=t_\mathrm{algo})\right),
\end{equation}
will consist of some amplitude damping channels that are degradable and some others that are anti-degradable. In this sense, additivity of the channel coherent information is assured for such a combination of channels:
\begin{itemize}
\item Degradable channels fulfill additivity, i.e. $Q_\mathrm{coh}(\mathcal{N}\otimes\mathcal{M})= Q_\mathrm{coh}(\mathcal{N})+Q_\mathrm{coh}(\mathcal{M})$, when both $\mathcal{N}$ and $\mathcal{M}$ are degradable \cite{wildeQIT}.
\item Degradable and anti-degradable channels fulfill additivity, i.e. $Q_\mathrm{coh}(\mathcal{N}\otimes\mathcal{M}) = Q_\mathrm{coh}(\mathcal{N})+Q_\mathrm{coh}(\mathcal{M})$, when $\mathcal{N}$ is degradable and $\mathcal{M}$ is anti-degradable \cite{bothAdd,degantidegAdd}.
\item Anti-degradable channels fulfill additivity, i.e. $Q_\mathrm{coh}(\mathcal{N}\otimes\mathcal{M})= Q_\mathrm{coh}(\mathcal{N})+Q_\mathrm{coh}(\mathcal{M})$, when both $\mathcal{N}$ and $\mathcal{M}$ are anti-degradable \cite{bothAdd,antiantiAdd}.
\end{itemize}

Consequently, the following holds for FTVAD channels
\begin{equation}\label{eq:proofADequal}
\begin{split}
&\lim_{n\rightarrow \infty}\frac{1}{n}Q_\mathrm{coh}\left(\bigotimes_{j=1}^n \mathcal{N}(\rho,t_1^{m^j},t=t_\mathrm{algo})\right) \\& = \lim_{n\rightarrow \infty}\frac{1}{n}\sum_{k=1}^n Q_\mathrm{coh}\left(\mathcal{N}(\rho,t_1^{m^j},t=t_\mathrm{algo})\right).
\end{split}
\end{equation}

Finally, considering the discussion from before (See equation \eqref{eq:proofConv}), we can conclude that the quantum capacity is actually equal to the ergodic quantum capacity for FTVAD channels:

\begin{equation}\label{eq:CergAD2}
\begin{split}
C_\mathrm{Q}(\bar{\gamma})=C_\mathrm{Q}^\mathrm{erg}(\bar{\gamma}) &= \mathrm{E}\{Q_\mathrm{coh}(\omega)\} \\&= \int Q_\mathrm{coh}(\gamma) p_{\gamma}(\gamma) d\gamma \\
&= \int Q_\mathrm{coh}(t_\mathrm{algo},t_1)p_{T_1}(t_1) dt_1,
\end{split}
\end{equation}
where in the last step we use the fact that the coherent information for AD channels is a function of the relaxation time $T_1$, which will be the random variable, as well as of the error correction cycle time (algorithm time), $t_\mathrm{algo}$. $\bar{\gamma}$ refers to the damping probability defined by the mean relaxation time, $\mu_{T_1}$ for the cycle time in consideration.

\end{proof}

Figure \ref{fig:CergAD} shows the quantum capacity of the fast TVAD as a function of the coefficient of variation of the relaxation time. In this figure, it can be observed that the capacity of the channel changes as a function of the coefficient of variation of the relaxation time. For low coefficients of variation ($\approx 1\%$), the difference between the capacity of the FTVAD channel and the quantum capacity of the static AD channel is negligible. When the coefficient of variation increases to about $c_\mathrm{v}\approx 20\%$ differences in the capacity of the FTVAD channel and the quantum capacity of the static AD channel become greater. Note also that for very noisy scenarios $\bar{\gamma}>0.35$, the quantum capacity of the FTVAD channel is higher than the quantum capacity of the static AD channel. This is a consequence of the fact that even if the mean value of the damping probability is the anti-degradable region ($\bar{\gamma}\in (1/2,1]$) of the AD channel, the fluctuations of the parameters imply that some of the actual $\gamma$ will lay in the degradable region \textbf{($\gamma\in [0,1/2]$)}.

We wrap up this discussion by analyzing the capacity of the twirled approximated versions of the TVAPD and TVAD channels. Similarly to the combined amplitude and phase damping channel, the quantum capacity of the static versions of these channels is not known, as they belong to the family of Pauli channels, which are not degradable. Actually, these channels have been proven to have strictly superadditive coherent information for some very noisy scenarios \cite{wildeQIT} which makes it impossible to reduce the calculation of the quantum capacity to a single-key formula. For this reason, the so-called hashing bound (the hashing bound is equivalent to bounding the capacity with the coherent information of the channel) is used as a good lower bound on the capacity of the aforementioned channels \cite{wildeQIT}. Thus, we can derive a lower bound for the ergodic quantum capacity of the fast time-varying twirled approximated channels as
\begin{equation}\label{eq:ChashergAD}
\begin{split}
C_\mathrm{Q}(\mathbf{\bar{p}}) &\geq C_\mathrm{H}^\mathrm{erg}(\mathbf{\bar{p}}) \\&= \mathrm{E}\{C_\mathrm{H}(\omega)\} \\&= \int C_\mathrm{H}(\mathbf{p}) p_{\mathbf{p}}(\mathbf{p}) d\mathbf{p} \\
&= \int\int C_\mathrm{H}(t_\mathrm{algo},t_1,t_2)p_{T_1,T_2}(t_1,t_2) dt_1 dt_2,
\end{split}
\end{equation}
where $\mathbf{p}$ refers to the array $(p_\mathrm{x},p_\mathrm{y},p_\mathrm{z})$ and $\mathbf{\bar{p}}$ refers to a similar vector derived considering the mean relaxation and dephasing times. Note that here $C_\mathrm{H}(\mathbf{\bar{p}}) = Q_\mathrm{coh}(\mathbf{\bar{p}})$. The expression given in \eqref{eq:ChashergAD} is reduced to a single integral that depends only on the relaxation time if the twirled versions of the amplitude damping channel are considered. We will refer to this lower bound as the ergodic hashing bound, $C_\mathrm{H}^\mathrm{erg}$, in order to be consistent with the terminology used for static Pauli channels. Figure \ref{fig:CerghashAD} presents the ergodic hashing bounds of the twirled approximated channels of the amplitude damping channel. These results show how the ergodic hashing bounds deviate from the static hashing bound when the coefficient of variation of the relaxation time increases and how both metrics coincide when $c_\mathrm{v}(T_1)\approx 1\%$.

\subsection{Performance of QECCs over FTVQCs}\label{sub:fastQECCs}
The performance of QECCs has been shown to worsen significantly when operating over slow time-varying quantum channels \cite{TVQC,outage}. In this section, we present the outcomes of simulations we have conducted to study the performance of QTCs and planar codes when their qubits are subjected to the effects of the FTVQC (See Appendix \ref{appB} for the details of the Monte Carlo numerical simulations). To conduct the simulations, we use slow and fast time-varying amplitude damping Clifford twirl approximations, STVADCTA and FTVADCTA respectively, whose static counterpart is the depolarizing channel. The asymptotic limits for error correction associated to these noise operations will be the hashing bound and the ergodic hashing bound, respectively. In our simulations we consider a coefficient of variation of $c_\mathrm{v}(T_1)=25\%$, a value that is typical in experimental superconducting qubits \cite{decoherenceBenchmarking,TVQC}.

Figure \ref{fig:FFQTC} shows the simulation outcomes obtained for the rate $1/9$ quantum turbo code when it operates over static channels, STVQCs and FTVQCs.

\begin{figure}[!ht]
\centering
\includegraphics[width=\linewidth]{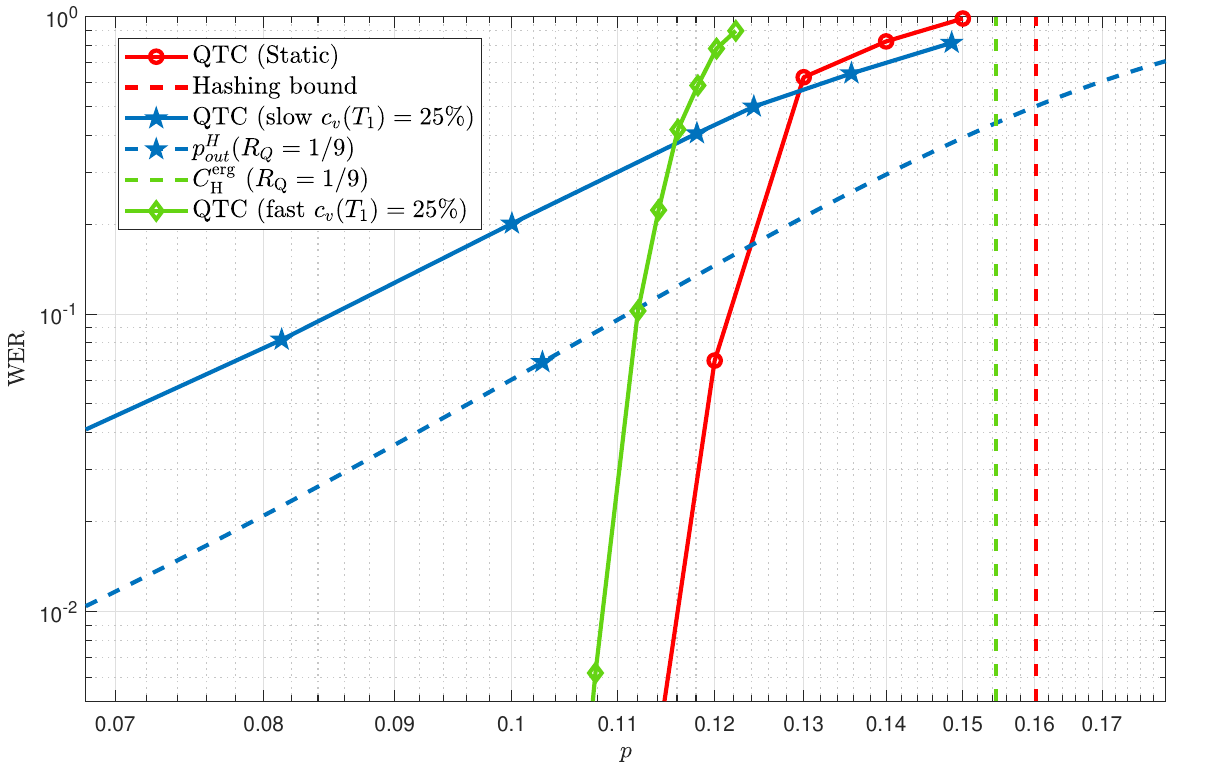}
\caption{\textbf{QTC operating over static, STVQC and FTVQC.} The coefficient of variation considered is $c_\mathrm{v}(T_1)=25\%$, which is a value that is typical in superconducting qubits \cite{decoherenceBenchmarking,TVQC}. \textbf{a} (red) static performance (solid) and hashing bound (dashed), $p^*=0.1602$ \textbf{b} (blue) STVQC performance (solid) and quantum outage probability \cite{outage} (dashed) \textbf{c} (green) FTVQC performance (solid) and ergodic hashing bound (dashed), $p^*_{\mathrm{erg}}=0.1545$.}
\label{fig:FFQTC}
\end{figure}

The results in Figure \ref{fig:FFQTC} show how the performance of the QTC, assessed in terms of of the Word Error Rate (WER), is worse over the FTVQC than over its static couterparts. In fact, as seen by comparing the hashing bound to the ergodic hashing bound, the loss in code performance is similar to the loss in quantum capacity. It should be noted that the flattening of the performance curve that QTCs suffered over STVQCs is not observed over FTVQCs. In fact, when comparing the performance curves obtained for the FTVQC channels and those derived for the static channels, the only difference is that the operation point of the code, the point where the waterfall region starts, is worse (the curves themselves have the same shape). Thus, even though the fluctuations of the decoherence parameters also worsen the performance of the code, this loss is much less significant if the fluctuations are local to each qubit of the system (as is the case with the qubits that are considered in this paper) than if they are completely correlated.

To provide further context, we have also studied the performance of $d\in\{3,5,7,9\}$ planar codes over the fast multi-qubit TVQC proposed in this paper. These results, along with the performance of the planar codes over static channels, are presented in Figure \ref{fig:FFplanar}. For the sake of clarity, in this figure we have omitted the performance results of these codes when they operate over STVQCs \cite{TVQC}.

\begin{figure}[!ht]
\centering
\includegraphics[width=\linewidth]{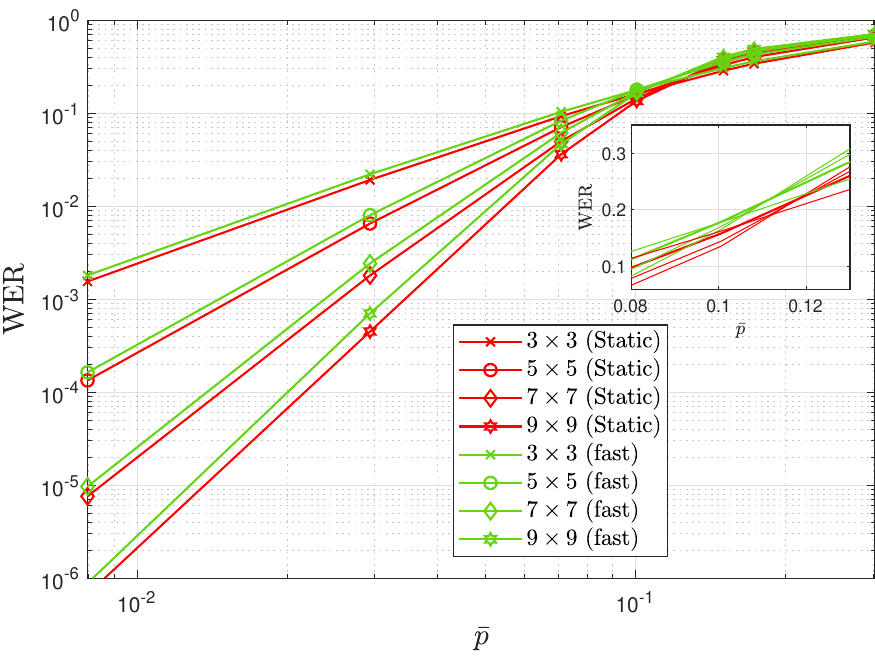}
\caption{\textbf{Planar codes operating over FTVQCs and static channels.} The coefficient of variation considered is $c_\mathrm{v}(T_1)=25\%$, which is a value that is typical in superconducting qubits \cite{decoherenceBenchmarking,TVQC}. A zoom to the code threshold region is also presented. \textbf{a} (red) Planar codes over static noise \textbf{b} (green) Planar codes over FTVQCs.}
\label{fig:FFplanar}
\end{figure}

Figure \ref{fig:FFplanar} shows how, in a similar manner to what was observed for QTCs, the performance of the planar codes is slightly worse over the FTVQCs than when considering time-invariant noise. Notably, this decrease in performance is less significant than for QTCs, but this is reasonable, since the performance of these planar codes over the static scenario is much worse than that of the turbo codes in this same context. Against this backdrop, the performance of planar codes can be said to deteriorate when the fluctuations of the decoherence parameters are considered. However, if the fluctuations are local to each of the qubits of the system, the loss in performance will not be catastrophic (a phenomenon that does actually occur over the STVQCs \cite{TVQC}).

Generally, surface codes are benchmarked based on a metric known as the code threshold, $p_\mathrm{th}$, which is defined as the physical error probability below which increasing the code distance actually implies lowering the $\mathrm{WER}$. In this sense, the surface code threshold is a capacity-like metric (the fact that operating over noise levels above the threshold does not make sense from the coding point of view is reminiscent of operating at rates higher than the capacity) that serves to assess the true error correction potential of this family of codes. The close-up image shown on the top-right of Figure \ref{fig:FFplanar} presents the threshold for planar codes both for static and fast time-varying scenarios. It can be observed that the planar code threshold is $p_\mathrm{th}\approx 0.112$ when the static noise model is considered, while it degrades to $p_\mathrm{th}\approx 0.105$ when the codes operate over the FTVADCTA with $c_\mathrm{v}(T_1)=25\%$. This slight degradation of the code threshold is in line with the slight decrease in quantum capacity discussed previously.

\section{Discussion}

In this paper, we have discussed the way multi-qubit time-varying quantum channels are constructed. In previous works, the fluctuations of the qubit decoherence parameters that describe multi-qubit TVQCs were assumed to be qubit-wise perfectly correlated in a block, and independent from block-to-block \cite{TVQC,outage}. However, recent experimental results have shown that the fluctuations of the $T_1$ and $T_2$ parameters of superconducting qubits are caused by the coupling of the qubit with an ensemble of enviromental unstable near-resonant two-level systems that arise from atomic-scale defects in the Josephson Junctions that make up these supeconducting qubits \cite{decoherenceBenchmarking,klimov,fluctAPS,SchlorPhD,TLSdefects,TLSphysrevB}. This means that the origin of the decoherence parameter fluctuations is local to each particular superconducting qubit in the system, granted that there is no coupling between the defects. The fact that the individual-qubit decoherence parameter fluctuations are uncorrelated was proven in \cite{decoherenceBenchmarking} for a system made up of $2$ superconducting qubits. In this work we extend this study to quantum hardware made up of $5$ superconducting qubits by concurrently and repeatedly measuring the relaxation times of the qubits. To do so, we have used the quantum processors ibmq\_quito, ibmq\_belem, ibmq\_lima, ibmq\_santiago and ibmq\_bogota, which can be accessed through the cloud using the IBM Quantum Lab platform (See Appendix \ref{appC} for details about the experiments). We have studied the correlation that exists between various measurements of $T_1$ values for each of the qubits in different operational scenarios. Our results show that for most of these scenarios the correlation coefficients are not significant enough to indicate any correlation. Thus, we have provided further evidence to support the claims that fluctuations of the decoherence parameters are local to each of the superconducting qubits for the superconducting quantum hardware considered here. Based on previous research and our own results, it is clear that considering the fluctuations of the decoherence parameters of superconducting qubits to be local to each particular qubit of the system is a reasonable assumption.

Earlier, we also discussed different ways of constructing multi-qubit time-varying quantum channels depending on the qubit-wise correlation of the fluctuations of $T_1$ and $T_2$. We analyzed the mathematical formalism of the previously considered multi-qubit TVQCs \cite{TVQC,outage}, which we named slow TVQCs due to their similarity with classical slow fading channels, and we proposed fast TVQCs as the appropriate way to construct multi-qubit decoherence models with parameter fluctuations that are local to each qubit. We saw how this phenomenon of local fluctuations is present in superconducting hardware. Moreover, and due to their similarity with classical fading channels, we discussed the quantum capacity of FTVQCs, and proposed the ergodic quantum capacity as a lower bound on the asymptotically achievable error rate for QEC for these noise models. Moreover, we proved that the ergodic quantum capacity coincides with the quantum capacity for the family of fast time-varying amplitude damping channels. We computed the ergodic quantum capacity numerically for fast time-varying amplitude damping channels and concluded that the loss in capacity caused by decoherence parameter fluctuations is small, similar to what happens in classical fast fading. Finally, we discussed the performance of quantum error correction codes when the noise operation is defined as a FTVQC by conducting numerical simulations of the performance of quantum turbo codes and planar codes. Interestingly, our results indicated that the word error rate of the codes worsens only slightly, similar to what happens to the quantum capacity, and in stark contrast to the drastic deterioration that the performance of these codes suffers over the previously considered STVQCs, where their WER curves flatten substantially. In summary, when operating over FTVQCs the shape of the performance curve of a QECC is the same as in the conventional noise model in the literature (all qubits suffer the same static noise), but its operating point is defined for lower physical error probabilities than for the static case. We also studied the threshold of the considered planar codes and we observed how it was slightly lowered over FTVQCs, similar to what happens to the quantum capacity.

Another important matter that we discussed is that understanding the nature of the qubit-to-qubit relationship of the fluctuations of the decoherence parameters of superconducting qubits is vital in order to characterize the real quantum noise that affects the quantum information that they encode. The slow TVQCs that have been presented in the literature \cite{TVQC,outage} predict a substantial QECC performance loss due to the time-varying behaviour of $T_1$ and $T_2$. However, we have seen that a silver lining exists when building QEC codes in superconducting hardware. Although superconducting qubits present substantial ($c_\mathrm{v} \gtrapprox 20\%$) parameter fluctuations, these fluctuations are local to each of the elements of the processor and, thus, are uncorrelated between qubits. While it is true that a slight performance loss is inevitable, it is significantly milder than the deterioration predicted by the STVQCs which assume that the fluctuations are fully qubit-wise correlated for a QEC block. To provide specific examples, the QTC considered in this article operates over FTVQCs at a physical error rate $\approx 7\%$ lower than over static channels for a $\mathrm{WER} = 10^{-2}$, and the threshold of the planar codes will be $\approx 6\%$ lower over FTVQCs than over static channels. These losses, despite being important, are much less restrictive than the flattening effect that takes place when fully correlated fluctuations are considered \cite{TVQC,outage}. These results imply that constructing qubits whose qubit-to-TLS defect interactions are local and uncorrelated with other qubits is critical to maintaining the excellent performance of QECCs when the decoherence parameters of the superconducting qubits fluctuate over time. The coupling between qubits, and the TLS defects themselves, are directly dependent on the architecture of each particular quantum chip. Thus, the results we present herein prove that experimentalists must consider these effects when designing hardware, as they will play a pivotal role in suppressing the amount of errors the hardware suffers. It must also be mentioned that minimizing the decoherence parameter fluctuations of superconducting qubits will also be important to obtain error correction codes that perform as well as it would be expected based on the results obtained using the static noise model prevalent in the literature. Another important research topic is to study the time-varying behaviour of the decoherence parameters that other qubit technologies (such as trapped ions, NV center qubits or silicon spin qubits) experience. This would allow us to understand if their noise dynamics are more accurately described by the time-varying quantum channels discussed in this article or by the traditional static noise models.

It is also critical to further study the fluctuations of the decoherence parameters that superconducting qubits exhibit. It must be noted that interest in these effects has recently picked up, especially in the experimental research community \cite{decoherenceBenchmarking,klimov,fluctAPS,SchlorPhD,TLSdefects,TLSphysrevB}. Nonetheless, more work on this topic is needed to completely understand the time-fluctuating behaviour of superconducting qubits and its causes. This will enable the creation of an accurate theoretical model and, ultimately, to mitigate the impact of time-dependant noise on quantum information. For example, in this article, we have studied the correlations that exist between the fluctuations of the qubits of some $5$ qubit quantum processors and we have concluded that there is no correlation significant enough to warrant the application of slow multi-qubit time-varying noise models. However, this might not be the case for other quantum processors that might have other architectures or may be comprised of more qubits. Note that the absorption of high energy particles can also generate correlated errors in superconducting devices \cite{gamma1,gamma2,gamma3}. Additionally, we must also mention that the objective of our experiments was not to obtain very accurate values of the correlation coefficients. Instead, our goal was to verify that, for the typical case, assuming that fluctuations are qubit-to-qubit uncorrelated is grounded on and backed up by experimental results (note that although the obtained $95\%$ confidence intervals are wide, they all contain negligible correlation values and prove that fluctuations are qubit-to-qubit uncorrelated). In any case, more accurate experiments on this topic should be conducted, as they will serve to better understand the dynamics of multi-qubit superconducting quantum processors. If significant correlation values are observed in future experiments, it may become necessary to invoke TVQCs with finite correlation (perhaps by modelling these events with Hidden Markov Models) to better represent multi-qubit time-varying noise. Another phenomenon that must be further explored is the sharp transition of the relaxation times that our results show for particular qubits (See Appendix \ref{appC}). This behaviour must be better understood and possibly included in the noise models if it is shown to be common for superconducting qubits. Furthermore, note that we have considered that the decoherence parameters of all the qubits have the same mean and standard deviation for the derivation of the capacity and the numerical simulations. This assumption was necessary to discuss the topics that comprise this paper, but the most accurate FTVQC for a multi-qubit superconducting channel will most likely have a set of parameters for each of the qubits of the processor (at least for the NISQ devices that exist at the time of writing).

Even though we have only studied the impact of parameter fluctuations from the point of view of quantum error correction, the time-varying decoherence models discussed in this article will also have an impact on near term quantum algorithm implementations on superconducting NISQ devices and on the error mitigation techniques used to ``clean'' the outcomes of NISQ devices. When implementing quantum algorithms in NISQ processors, the physical qubits that execute the operations of these algorithms will be subjected to decoherence (as well as gate and measurement errors) that will corrupt the desired outcomes. So far we have seen how the decoherence mechanisms that affect superconducting qubits are better described by FTVQCs than by STVQCs or static channels. One way to obtain the ``best'' version of quantum algorithms in the presence of fast varying quantum noise might be to allocate their resources (qubits) as functions of the noise itself (some qubits might be able to perform longer-lasting tasks than others). This might allow quantum software developers to determine exactly what limits the algorithm (how many gates can be applied before there is too much noise). Another possibility might be to apply error mitigation techniques based on the calibration of the device in order to post-process results and derive more accurate outputs than the raw yield of the NISQ device. However, these techniques and simulations strongly depend on the calibration data of the system. For this reason, fluctuations of decoherence parameters must also be taken into account for all these design and simulation tasks if they are be run in NISQ devices based on superconducting technologies. In this way, it will become possible to build better mitigation techniques and more noise resilient NISQ algorithms.

Regarding quantum information theory, additional work is still required to fully understand the behaviour of the quantum capacity of the channels proposed in this article. Although we have been able to prove that the ergodic quantum capacity is actually the quantum capacity for the fast time-varying amplitude damping channels, we have just lower bounded this quantity for the more general fast time-varying combined amplitude and phase damping channel that includes non-negligible pure dephasing channels. Since most of the superconducting qubits that exist in the literature do not saturate the Ramsey limit (the Ramsey limit refers to the scenario $T_2 \approx 2T_1$), their dynamics (including pure dephasing) are described by the latter channel. However, as mentioned previously, it is not known if the coherent information of the combined amplitude and phase damping channel is additive. Thus, it may be that this quantity is superadditive, which implies that our understanding of this topic should increase before the capacity of the FTVQC version of this channel can be studied. This is especially relevant, since including the fluctuations of the pure dephasing time in a quantum noise model will worsen the capacity more than when only $T_1$ is considered. Additionally, the time-varying quantum channel models discussed herein might also be adapted to other channels (aside from the family of amplitude damping and dephasing channels considered in this article) if the noise parameters that define them also present behaviour similar to the relaxation and dephasing times of the superconducting qubits.

All in all, we consider that the FTVQC model proposed in this article describes the noise suffered by superconducting multi-qubit systems more accurately than previously considered noise models, at least for the hardware considered in this article. This claim is backed up by the experiments we have conducted and by previous literature on qubit-to-TLS defect interactions. Consequently, we expect that quantum error correction codes implemented in these types of quantum processors will perform worse than what would be expected based on results obtained for static quantum channels. Once more, it is necessary for more research on the topic of decoherence parameter fluctuation and its incorporation to the decoherence models to be conducted in order to unveil the true performance of near- and long-term quantum error correction codes as well as NISQ algorithms and error mitigation techniques.

\section*{Acknowledgements}
We acknowledge the use of IBM Quantum services for this work. The views expressed are those of the authors, and do not reflect the official policy or position of IBM or the IBM Quantum team. We also want to acknowledge Luis Vitores and I\~nigo Barasoain for useful discussions regarding hypothesis testing and bootstrapping. We also want to acknowledge Bruno Sunsundegi for numerically studying the $3\times 3$ planar code operating over the FTVQC as part of his bachelor's thesis.

This work was supported by the Spanish Ministry of Economy and Competitiveness through the ADELE project (Grant No. PID2019-104958RB-C44), by the Spanish Ministry of Science and Innovation through the proyect Few-qubit quantum hardware, algorithms and codes, on photonic and solid-state systems (PLEC2021-008251), by the Diputación Foral de Gipuzkoa through the DECALOQC Project No. E 190 / 2021 (ES), by the Ministry of Economic Affairs and Digital Transformation of the Spanish Government through the QUANTUM ENIA project call - QUANTUM SPAIN project, and by the European Union through the Recovery, Transformation and Resilience Plan - NextGenerationEU within the framework of the Digital Spain 2025 Agenda. This work was funded in part by NSF Award No. CCF-2007689.

\section*{Author Contributions}
J.E.M. conceived the research. J.E.M., P.F. and J.R.F. discussed and implemented the experiments in the IBM machines. J.E.M. constructed the models. J.E.M., P.M.C. and J.G.-F. proposed the asymptotical limits. J.E.M. and A.dM.iO. performed the numerical simulations. J.E.M., P.F. and A.deM.iO. analyzed the results and drew the conclusions. The manuscript was written by J.E.M., P.F. and A.deM.iO., and revised by P.M.C., J.G.-F and J.R.F. The project was supervised by J.E.M., P.M.C., J.G.-F. and J.R.F.

\appendix
\section{Pearson correlation tests}\label{appA}

In section \ref{res:exp} we studied the statistical dependence between the fluctuations of the relaxation times of the qubits of various different IBM quantum processors. The objective of this analysis was to determine if the aforementioned fluctuations were local to each particular qubit. To do so we used the Pearson correlation coefficient, as it provides a measure of the correlation that exists between measured sequences \cite{PearsonBounds,strengthPear,freedman}. For a pair of random variables $(\mathrm{X},\mathrm{Y})$, the sample Pearson correlation coefficient, $r_{\mathrm{XY}}$, is defined as  \cite{PearsonBounds}
\begin{equation}\label{eq:pearson}
\begin{split}
r_{\mathrm{XY}} &= \frac{\mathrm{cov}(\mathrm{X},\mathrm{Y})}{\sigma_\mathrm{X} \sigma_\mathrm{Y}} \\&= \frac{n\sum_{i=1}^{n}x_iy_i - \left(\sum_{i=1}^n x_i\right)\left(\sum_{i=1}^n y_i\right)}{\sqrt{n\sum_{i=1}^{n}x_i^2 - \left(\sum_{i=1}^n x_i\right)^2} \sqrt{n\sum_{i=1}^{n}y_i^2 - \left(\sum_{i=1}^n y_i\right)^2}},
\end{split}
\end{equation}
where $\mathrm{cov}(\mathrm{X},\mathrm{Y})$ denotes the covariance and $\sigma$ refers to the standard deviation.

In the main text, we calculated the correlation coefficient of the measured qubit relaxation times for each of the studied quantum processors. To ensure the statistical significance of the obtained correlation coefficients, we calculate the $95\%$ confidence intervals via bootstrapping \cite{bootstrap}. Bootstrapping is a method that uses random resampling and replacement of samples to mimic the original population from which the samples were extracted. The bootstrap probability distribution can then be used to derive a significance confidence interval. For the analysis we have conducted in this paper, we obtain the $95\%$ confidence interval as the interval that encompasses the $2.5^{th}$ and the $97.5^{th}$ percentiles of the resampled Pearson correlation coefficient values. In this way, we can be $95\%$ confident that the correlation parameter that exists between those random variables will fall within said confidence interval. Confidence intervals can also be used to reject or retain the null hypothesis of a hypothesis test \cite{CIhypothesis}. Note that here we are actually performing a hypothesis test to determine if the variables are fully correlated (null hypothesis) or if they show some degree of correlation (alternate hypothesis). Thus, if the null hypothesis lays in the derived confidence interval, we cannot exclude it as being the population parameter at the chosen level of confidence.

The final component of the statistical dependence analysis we conduct herein is an accurate interpretation of the obtained values. This allows us to determine if significant correlations exist or not. The ranges of values for which two parameters might be strongly correlated depends on the actual problem (field). However, for physical sciences, there should be no doubt about the dependence between two variables, implying that strong or significant correlation values should be high ($|r_{\mathrm{XY}}| \gtrapprox 0.9$) \cite{PearsonBounds,strengthPear}. For low values of $r_{\mathrm{XY}}$, no considerable dependence (weak) relationship can be concluded. Specifically, for classical fading channels, whenever the spatial correlation coefficient values are approximately lower than $0.6$, it is often assumed that the fading gains of the Rayleigh channel are i.i.d. since both channels will be very similar, i.e., the correlation is negligible \cite{classicalCorr1,classicalCorr2,classicalCorr3}.

\section{QECC numerical simulation}\label{appB}
Monte Carlo computer simulations of the $d\times d$ planar codes with $d\in\{3,5,7,9\}$ \cite{surfaceFowler,surface} and of the QTC of rate $1/9$ in \cite{josu2} have been carried out to estimate changes in their performance over various different operational scenarios. Planar codes belong to the more general family of surface codes \cite{surfaceFowler,surface} and are $[[d^2+(d-1)^2,1,d]]$ QECCs defined by the grid length of the code $d$. A blocklength of $k=1000$ logical qubits has been selected for the QTC, as in \cite{josu,josu2}.

Planar codes are decoded using a Minimum Weight Perfect Matching (MWPM) decoder, which is implemented using the QECSIM tool \cite{surface}. The QTC is decoded via the decoding algorithm presented in \cite{QTC,EAQTC}, which combines two Soft-In Soft-Out (SISO) decoders.

Each round of the numerical simulation is performed by generating an $N$-qubit Pauli operator, calculating its associated syndrome, and finally running the decoding algorithm using the syndrome as its input. Once the logical error is estimated, it is compared with the channel error in order to decide if the decoding round was succesful. The operational figure of merit we use to evaluate the performance of these quantum error correction schemes is the Word Error Rate ($\mathrm{WER}$). The WER represents the probability that at least one qubit of the received block has been incorrectly decoded.

For the numerical Monte Carlo methods employed to estimate the $\mathrm{WER}$ of the Kitaev toric codes and the QTC, we have applied the following rule of thumb to select the number of blocks to be transmitted, $N_{\mathrm{blocks}}$ \cite{josu,josu2}, as
\begin{equation}
N_{\mathrm{blocks}} = \frac{100}{\mathrm{WER}}.
\end{equation}
As explained in \cite{josu,josu2}, under the assumption that the observed error events are independent, this results in a $95\%$ confidence interval of about $(0.8\hat{\mathrm{WER}} , 1.25\hat{\mathrm{WER}})$, where $\hat{\mathrm{WER}}$ refers to the empirically estimated value for the $\mathrm{WER}$.

\section{Intra-calibration decoherence parameter fluctuation for the qubits of IBM quantum hardware}\label{appC}
In the main article, we discussed the fact that the fluctuations of the decoherence parameters of superconducting qubits are local, that is, the random variables $T_1(\omega)$ and $T_2(\omega)$ are qubit-wise uncorrelated. In \cite{decoherenceBenchmarking}, the authors proved this to be true for the relaxation time fluctuations by using their two-qubit superconducting system. In this appendix, we perform a similar analysis for five IBM $5$-qubit superconducting processors that are accessible online: ibmq\_quito, ibmq\_belem, ibmq\_lima, ibmq\_santiago and ibmq\_bogota \cite{IBMqexp}.

Qubit relaxation time, $T_1$, refers to the characteristic timescale at which a qubit in an excited state, $|1\rangle$, decays to its ground state, $|0\rangle$, caused by simple spontaneous emission. Consequently, the experiment that is usually performed in order to estimate the parameter $T_1$ of a qubit consists in collecting the statistics of the decay curve for the probability of measuring the excited state over time, $P_1(t)$. This is done by choosing a set of delay times $t_1,\cdots,t_n$ and then repeating the following protocol $N$ times for each of them \cite{experiments}:
\begin{itemize}
\item Prepare the qubit in the $|1\rangle$ state. This is usually done by exciting the qubit in the ground state via a Pauli $\mathrm{X}$ operator.
\item Wait a delay time, $t_j$.
\item Measure the qubit in the computational basis ($|0\rangle,|1\rangle$).
\end{itemize}

Once the decay curve, $P_1(t)$, is obtained, a fit to exponential decay is performed in order to estimate the value of the qubit relaxation time \cite{experiments}. As explained in the main text, we are interested in studying the locality of the fluctuations of the relaxation times of some of the IBM $5$-qubit quantum processors. Thus, we run the previously presented experiment over time, simultaneously for the all the qubits of the systems in question. Figure \ref{fig:exps} portrays the schematic of each of the experiments that we have conducted. We run each experiment $4000$ times ($4000$ shots) for $20$ uniformly separated delay times starting from a delay of $t_1=1\mu s$ to $t_{20}=2*T_1^{cal(i)}$, where by $T_1^{cal(i)}$ we refer to the relaxation time for the $i^{th}$ qubit provided by IBM for the specific calibration cycle in which the experiments are done. Note that the calibration data provided by IBM refers to measurements performed during that precise calibration round. However, these values will actually fluctuate within the calibration cycle itself, similarly to the superconducting qubits of \cite{decoherenceBenchmarking}. These intra-calibration fluctuations of the relaxation parameter are precisely what we are interested in observing.

\begin{figure}[h]
\centering
\leavevmode
\Qcircuit @C=1em @R=0.5em @!{
& \lstick{\ket{0}} &  \gate{X} & \gate{\mathrm{Delay}(t_j)} & \meter  \\
& \lstick{\ket{0}} &  \gate{X} & \gate{\mathrm{Delay}(t_j)} & \meter  \\
& \lstick{\ket{0}} &  \gate{X} & \gate{\mathrm{Delay}(t_j)} & \meter  \\
& \lstick{\ket{0}} &  \gate{X} & \gate{\mathrm{Delay}(t_j)} & \meter  \\
& \lstick{\ket{0}} &  \gate{X} & \gate{\mathrm{Delay}(t_j)} & \meter  \\
}
\caption{Schematic representation of each of the experiments done in the IBM quantum processors in order to estimate the decay curve of each of the qubits.}
\label{fig:exps}
\end{figure}
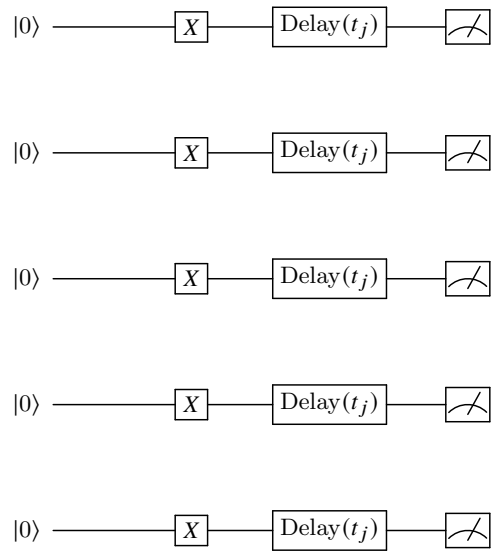

As discussed in the article, we have conducted these intra-calibration $T_1$ measurements for the ibmq\_quito, ibmq\_belem, ibmq\_lima, ibmq\_santiago and ibmq\_bogota $5$ qubit quantum processors on different days. In Table \ref{tab:experimentsSpecs}, we detail the information related to each of the scenarios that we have tested. Note that the number of measurements for each of the considered scenarios/processors, as well as the duration of the experiments themselves, is different. The reason for this is that the IBM machines are calibrated at different times with different frequencies. We observed that the ibmq\_santiago and ibmq\_bogota, which are of the Falcon r4L type, are more frequently (several times each day) calibrated than the ibmq\_quito, ibmq\_belem and ibmq\_lima processors (calibrated once a day approximately), which belong to the Falcon r4T class. Since we are interested in the intra-calibration fluctuations, the duration of those cycles is something that must be accounted for, which leads to some scenarios having more $T_1$ measurements than others. Another limitation that must be disclosed is the fact that these machines can be accessed by the general public, resulting in queues and wait times to run the experiments. Consequently, whenever a large number of tasks are sent to a machine (high demand for the processor in question), our experiments will be more spread out in time (implying that less measurements will be made in these specific calibration cycles). To be more specific, we have been able to run most of the experiments on days where the demand for the IBM systems was low, approximately running the measurement of the $T_1$ of the $5$ qubits of the system once every two minutes (with the exception of the ibmq\_santiago processor where the experiment was run once every four minutes). 

\begin{table*}[]
\caption{Relaxation time measurement specifications for each of the scenarios. The table includes the timestamp of the beginning of the experiments and their duration, as well as the number of measurements and the calibration data provided by IBM. The calibration relaxation times are provided in $\mu s$ units.}
\centering
\begin{tabular}{|l|l|l|l|l|l|l|l|l|l|}
\hline
\textbf{Scenario} & \textbf{Timestamp} & \textbf{\# Meas.} & $T_1^{cal(0)}$ & $T_1^{cal(1)}$ & $T_1^{cal(2)}$ & $T_1^{cal(3)}$ & $T_1^{cal(4)}$   \\ \hline
ibmq\_quito S1  &    15/04 00:00 (4 hours)         &         $213$           &     $61.208$      &      $92.812$            &   $156.228$  &  $90.232$ &                     $119.197$ \\ \hline 
ibmq\_quito S2  &        27/04 12:06 (11 hours)      &         $401$      &      $76.5$      &  $109.8$    &  $93.71$   &   $143.49$ &  $57.27$     \\ \hline
 ibmq\_belem S1   &    08/04 13:55 (24 hours)          &         $999$         &    $120.88$      &       $103.85$          &  $94.75$ &  $88.54$  &                            $83.65$ \\ \hline
ibmq\_belem S2    &   26/04 15:39 (7 hours)           &       $401$          &     $67.4$     &   $77.56$              &   $83.35$    & $86.44$  &    $64.2$  \\ \hline
  ibmq\_lima S1  &     15/04 14:21 (27 hours)       &       $1000$         &      $136.87$      &      $125.03$    &  $108.45$    & $126.02$   & $21.45$    \\ \hline
  ibmq\_lima S2  &     25/04 10:25 (6 hours)         &       $239$          &     $122.99$       &     $100.58$           &  $73.9$    & $97.38$   &               $22.74$ \\ \hline
  ibmq\_santiago S1  &   18/04 01:26 (12 hours)        &     $200$           &  $157.41$    &    $138.6$    &     $143.18$       & $118.58$  &              $166.08$ \\ \hline
  ibmq\_bogota S1  &     17/04 21:21 (4 hours)     &       $200$           &     $123.06$      &      $187.98$         &   $175.79$       &  $217.003$   &  $168.35$\\ \hline
\end{tabular}
\label{tab:experimentsSpecs}
\end{table*}

It is important to mention that conducting the relaxation time measurements every few minutes is not actually a problem for this study. The stochastic processes that typically define the fluctuations of $T_1$ were previously studied in \cite{decoherenceBenchmarking}. This work concluded that the stochastic process coherence times (the time for which the stochastic process can be considered to be approximately constant), $T_c$, are typically in the order of minutes \cite{TVQC}. Consequently, the stochastic process can be modeled as a random variable, $T_1(\omega)$, that is considered to be constant for a time $T_c$. Since we are interested in studying the correlation that exists between the random variables of $T_1$ for each of the qubits of a processor, the measurements must be performed sufficiently apart in time so that they do not belong to the same stochastic coherence period. As discussed before, our measurements are conducted several minutes apart from each others; hence, they are consistent with this reasoning.

\begin{table*}[]
\caption{Sample mean relaxation times, $\hat{\mu}$, sample standard deviation, $\hat{\sigma}$, and the coefficients of variation, $c_\mathrm{v}=\sigma/\mu$, for each of the qubits of the considered scenarios. Qubit $j$ is labelled by $T_1^{(j)}$.}
\centering
\setlength{\tabcolsep}{2pt}
\renewcommand{\arraystretch}{1.5} 
\begin{tabular}{|l|l|l|l|l|l|}
\hline
\textbf{Scenario} & $c_\mathrm{v}(T_1^{(0)})\left(\frac{\hat{\sigma}_{T_1^{(0)}}}{\hat{\mu}_{T_1^{(0)}}}\right)$ & $c_\mathrm{v}(T_1^{(1)})\left(\frac{\hat{\sigma}_{T_1^{(1)}}}{\hat{\mu}_{T_1^{(1)}}}\right)$ & $c_\mathrm{v}(T_1^{(2)})\left(\frac{\hat{\sigma}_{T_1^{(2)}}}{\hat{\mu}_{T_1^{(2)}}}\right)$ & $c_\mathrm{v}(T_1^{(3)})\left(\frac{\hat{\sigma}_{T_1^{(3)}}}{\hat{\mu}_{T_1^{(3)}}}\right)$ & $c_\mathrm{v}(T_1^{(4)})\left(\frac{\hat{\sigma}_{T_1^{(4)}}}{\hat{\mu}_{T_1^{(4)}}}\right)$   \\ \hline
ibmq\_quito S1  &    $15.5\% \left(\frac{11.23}{72.51}\right)$        &         $19.96\% \left(\frac{11.21}{56.15}\right)$           &     $22.44\% \left(\frac{21.02}{93.63}\right)$  & $12.76\% \left(\frac{12.42}{97.33}\right)$   &  $7\% \left(\frac{1.73}{24.762}\right)$  \\ \hline 
ibmq\_quito S2  &       $30.2\% \left(\frac{17.63}{58.37}\right)$        &         $15.62\% \left(\frac{10.02}{64.16}\right)$        &      $24.43\% \left(\frac{16.23}{66.46}\right)$ &  $26\% \left(\frac{15.61}{60.02}\right)$   & $9.39\% \left(\frac{5.4}{57.14}\right)$  \\ \hline
 ibmq\_belem S1   &    $13.76\% \left(\frac{12.86}{93.5}\right)$           &         $12.32\% \left(\frac{11.2}{90.87}\right)$           &    $7.88\% \left(\frac{6.73}{85.4}\right)$ & $11.71\% \left(\frac{12.21}{104.22}\right)$    & $20.14\% \left(\frac{13.54}{67.23}\right)$   \\ \hline
ibmq\_belem S2    &      $14.03\% \left(\frac{9.1}{64.66}\right)$        &       $14.93\% \left(\frac{9.26}{62.1}\right)$            &     $8.9\% \left(\frac{6.38}{71.78}\right)$ & $24.38\% \left(\frac{19.6}{80.41}\right)$ & $47.41\% \left(\frac{22.79}{48.1}\right)$      \\ \hline
  ibmq\_lima S1  &      $31.66\% \left(\frac{26.26}{82.94}\right)$       &      $9.94\% \left(\frac{11.11}{111.66}\right)$          &      $26\% \left(\frac{27.54}{105.85}\right)$  & $12.04\% \left(\frac{12.81}{106.33}\right)$  & $11.8\% \left(\frac{2.52}{21.36}\right)$   \\ \hline
  ibmq\_lima S2  &      $13.4\% \left(\frac{13.9}{111.25}\right)$       &       $11.97\% \left(\frac{12.93}{108.06}\right)$          &    $16.27\% \left(\frac{9.31}{57.27}\right)$ & $19.9\% \left(\frac{18.7}{93.91}\right)$ & $5.14\% \left(\frac{1.1}{21.52}\right)$ \\ \hline
  ibmq\_santiago S1  &     $14.97\% \left(\frac{21.03}{140.48}\right)$     &     $21.32\% \left(\frac{15.67}{73.47}\right)$           &  $22.67\% \left(\frac{32.85}{144.93}\right)$ &  $11.79\% \left(\frac{13.91}{117.98}\right)$ & $17.36\% \left(\frac{21.96}{126.56}\right)$\\ \hline
  ibmq\_bogota S1  &     $9.69\% \left(\frac{9.71}{100.1}\right)$     &       $12.91\% \left(\frac{17.2}{133.15}\right)$           &     $12.64\% \left(\frac{21.14}{167.21}\right)$  & $11.95\% \left(\frac{21.26}{177.76}\right)$ & $10.9\% \left(\frac{15.64}{143.42}\right)$ \\ \hline
\end{tabular}
\label{tab:cvsExperiments}
\end{table*}

Figures \ref{fig:quitoS1}-\ref{fig:bogotaS1} show the results of the relaxation time measurements of the qubits of the IBM quantum processors we have considered. The estimated mean relaxation times, as well as the estimated standard deviations and coefficients of variation($c_\mathrm{v}=\sigma/\mu$) are provided in Table \ref{tab:cvsExperiments}. It can be seen that the fluctuations exhibited by the relaxation times of the qubits of the systems are considerable, ranging from coefficients of variation of approximately $5\%$ up to even $47\%$. It is worth noting that the error bars (we actually do not plot these for the sake of clarity) are not beyond $5\%$ of the estimated data, indicating that the fluctuations are actually relevant and are not related to errors that may have arisen due to the fitting of the data to the relaxation decay curves \cite{fluctErrorBars}. We must also mention that the circuits we have built are also affected by gate and SPAM (state preparation and measurement) errors, which are unavoidable in current quantum hardware. However, our experiments contain a small number of gates ($5$ $\mathrm{X}$ Pauli gates, which present errors in the order of $10^{-4}$ \cite{IBMqexp}) and post-decay measurements (which present errors in the order of $10^{-2}$), and the contributions of these error sources are suppressed by the multiple executions we are doing for each delay of each of the experiment ($4000$ shots for each delay). Consequently, following the rationale in \cite{experiments}, we can say that even if errors of this type are present in the data, their impact on the results will be unimportant when enough shots are run.

As stated previously, the results obtained for the fluctuations of $T_1$ for each of the considered scenarios are presented in Figures \ref{fig:quitoS1}-\ref{fig:bogotaS1}. As expected, it can be observed that the fluctuations are notable for most of the studied cases. In fact, there are qubits in particular scenarios that exhibit sharp transitions in the level where the relaxation time fluctuates (See Figures \ref{fig:quitoS2}a, \ref{fig:belemS1}e and \ref{fig:belemS2}e). In these cases, the qubits manifest a step-like transition at a given point in time, at which point the relaxation time fluctuates around a different ``mean'' level. Notice how for scenario \ref{fig:belemS1}e, the relaxation time exhibits this effect before going back to its original level. The reason behind these sudden changes is unclear, but speculation regarding this topic is possible. Sharp drops in the relaxation times of qubits have previously been observed due to the absorption of high energy particle impacts by superconducting qubits \cite{gamma1,gamma2,gamma3}. However, such events result in correlated errors over the array of superconducting qubits (global drop of the $T_1$ of those qubits), an effect that has not been observed in our scenarios. We consider that further experimental study of these types of events is necessary to fully understand what causes such sharp drops of the relaxation times and their subsequent return to the original level.

We finish this discussion by analyzing the dephasing times of the superconducting qubits. To obtain values for the Ramsey dephasing times, $T_2^*$, we need to employ a similar procedure to the one used to measure the relaxation time, albeit with a circuit that includes a Hadamard gate (necessary to obtain a $|+\rangle$ state), the variable delay, and another Hadamard gate applied prior to the measurement itself. If there were no decoherence, then the obtained result would always be the $|0\rangle$ state (Hadamard gates are unitary), but dephasing (recall that $T_2^*$ is a combination of both relaxation and pure dephasing) will increase the probability of measuring the $|1\rangle$ state. This probability of measuring $1$ will have an exponentially decaying cosine shape \cite{experiments}, and then the obtained results can be fitted to such a function obtaining the Ramsey dephasing time. However, this is a significantly more nuanced experiment than the one used to measure the relaxation time, and measuring Ramsey times in IBM quantum processors has been shown to be unrealiable \cite{experiments}. Consequently, we will not conduct experiments to measure the Ramsey times. In any case, fluctuations of the dephasing time have also been experimentally observed \cite{decoherenceBenchmarking,fluctDeph} and assuming that they are local to each of the qubits seems reasonable (part of the dephasing phenomenon is caused by relaxation, which in itself is local to each qubit).

\begin{figure*}[!ht]
\centering
\includegraphics[width=\linewidth]{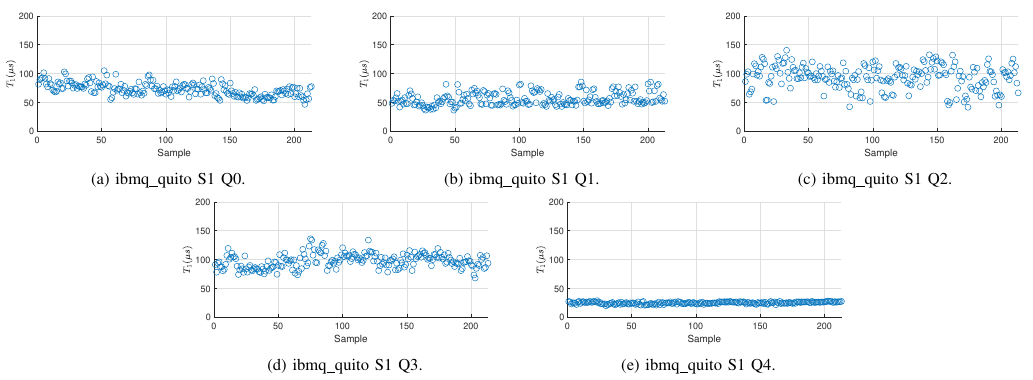}
\caption{$T_1$ measurements of the qubits of ibmq\_quito S1.}
\label{fig:quitoS1}
\end{figure*}
\begin{figure*}[!ht]
\centering
\includegraphics[width=\linewidth]{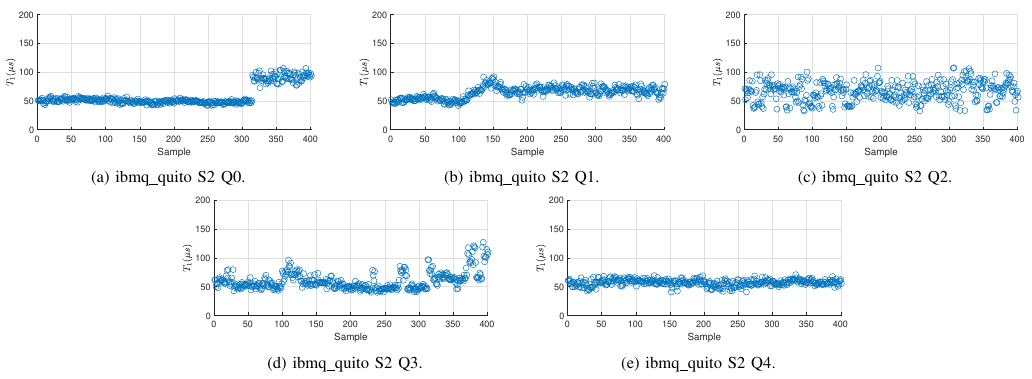}
\caption{$T_1$ measurements of the qubits of ibmq\_quito S2.}
\label{fig:quitoS2}
\end{figure*}
\begin{figure*}[!ht]
\centering
\includegraphics[width=\linewidth]{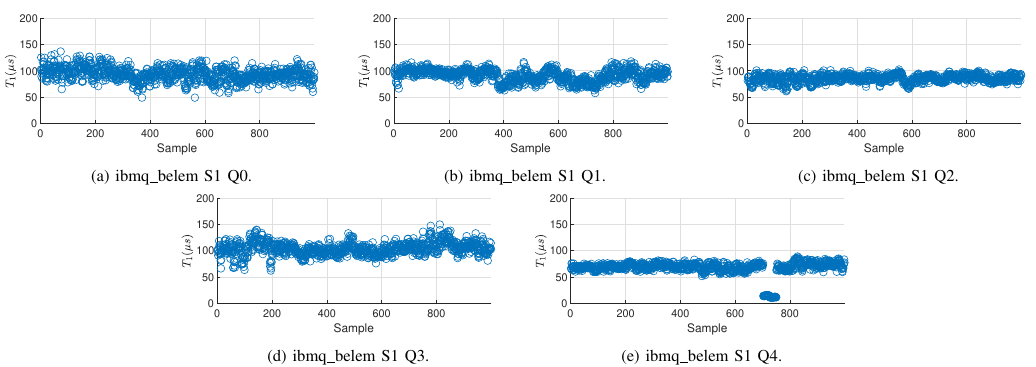}
\caption{$T_1$ measurements of the qubits of ibmq\_belem S1.}
\label{fig:belemS1}
\end{figure*}
\begin{figure*}[!ht]
\centering
\includegraphics[width=\linewidth]{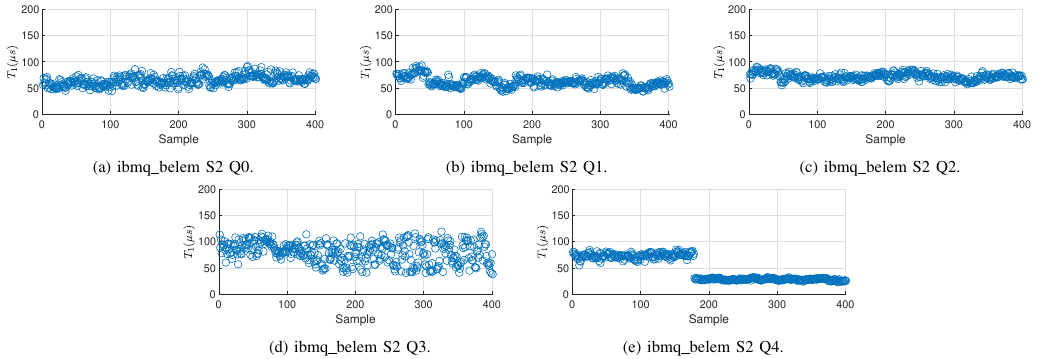}
\caption{$T_1$ measurements of the qubits of ibmq\_belem S2.}
\label{fig:belemS2}
\end{figure*}
\begin{figure*}[!ht]
\centering
\includegraphics[width=\linewidth]{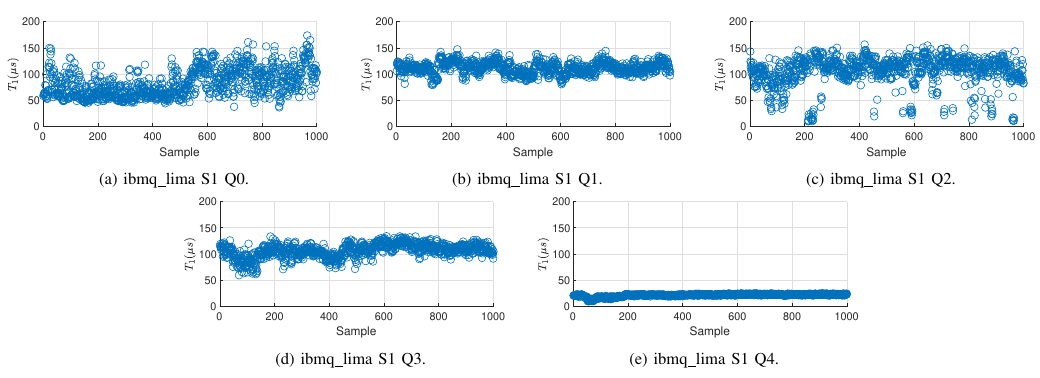}
\caption{$T_1$ measurements of the qubits of ibmq\_lima S1.}
\label{fig:limaS1}
\end{figure*}
\begin{figure*}[!ht]
\centering
\includegraphics[width=\linewidth]{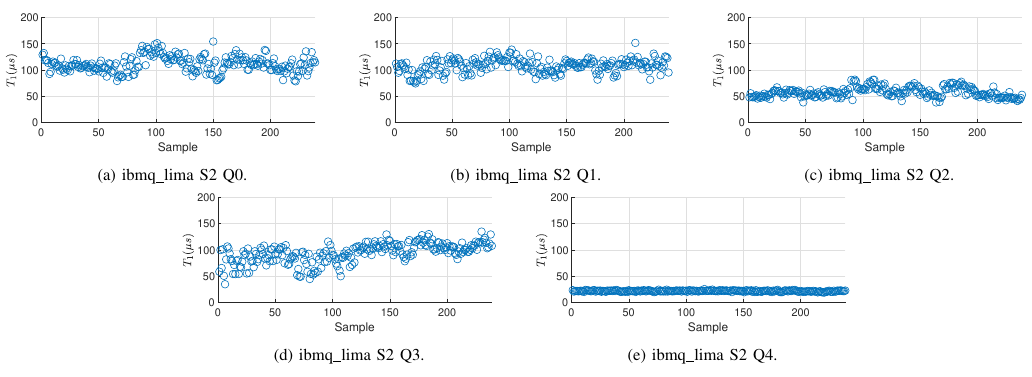}
\caption{$T_1$ measurements of the qubits of ibmq\_lima S2.}
\label{fig:limaS2}
\end{figure*}
\begin{figure*}[!ht]
\centering
\includegraphics[width=\linewidth]{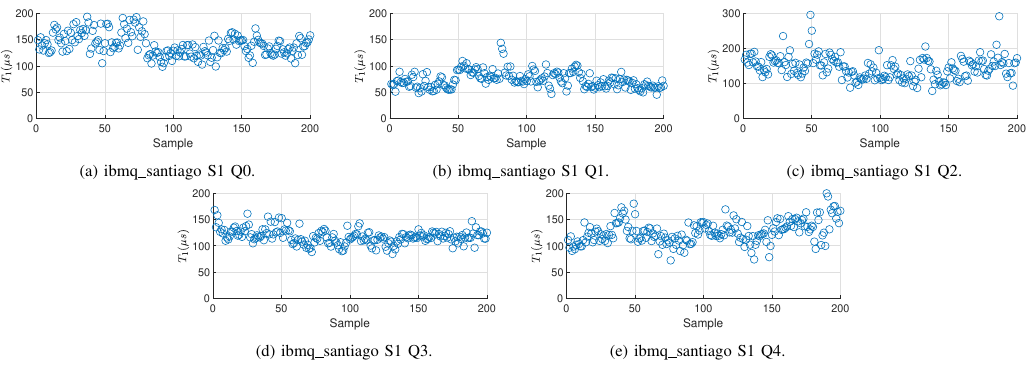}
\caption{$T_1$ measurements of the qubits of ibmq\_santiago S1.}
\label{fig:santiagoS1}
\end{figure*}
\begin{figure*}[!ht]
\centering
\includegraphics[width=\linewidth]{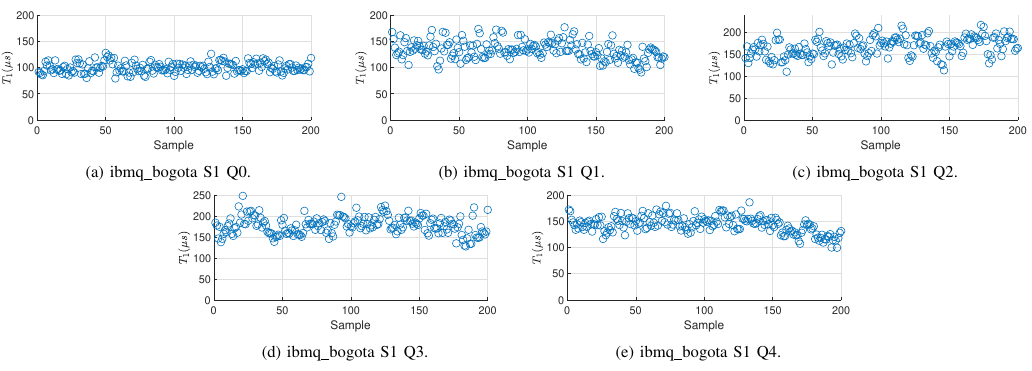}
\caption{$T_1$ measurements of the qubits of ibmq\_bogota S1.}
\label{fig:bogotaS1}
\end{figure*}


\begin{thebibliography}{00}

\bibitem{preskill}
Preskill, J. Quantum Computing in the NISQ era and beyond. \emph{Quantum} \textbf{2,} 79 (2018).

\bibitem{josuchannels} Etxezarreta Martinez, J., Fuentes, P., Crespo, P. M. \& Garcia-Fr\'ias, J. Approximating Decoherence Processes for the Design and Simulation of Quantum Error Correction Codes on Classical Computers. \emph{IEEE Access} \textbf{8,} 172623-172643 (2020).

\bibitem{decoherenceBenchmarking} Burnett, J. J. et al. Decoherence benchmarking of superconducting qubits. \emph{npj Quantum Inf.} \textbf{5}, 54 (2019).

\bibitem{klimov} Klimov, P. V. et al. Fluctuations of Energy-Relaxation Times in Superconducting Qubits. \emph{Phys. Rev. Lett.} \textbf{121}, 090502 (2018).

\bibitem{fluctAPS}  Schl\"or, S. et al. Correlating Decoherence in Transmon Qubits: Low Frequency Noise by Single Fluctuators \emph{Phys. Rev. Lett.} \textbf{123}, 190502 (2019).

\bibitem{fluctApp} Stehli, A. et al. Coherent superconducting qubits from a subtractive junction fabrication process. \emph{Appl. Phys. Lett.} \textbf{117}, 124005 (2020).

\bibitem{SchlorPhD}
Schl\"or, S. Intrinsic decoherence in superconducting quantum circuits. https://inis.iaea.org/search/search.aspx?orig\_q=RN:51064817 (2020).

\bibitem{wallraff} Krinner, S., Lacroix, N., Remm, A. et al. Realizing repeated quantum error correction in a distance-three surface code. \emph{Nature} \textbf{605,} 669–674 (2022).

\bibitem{bicycle} MacKay, D. J. C., Mitchinson, G. \& McFadden P. L. Sparse-graph
codes for quantum error correction. \emph{IEEE Trans. Inf. Theory} \textbf{50}, 2315-2330 (2004).

\bibitem{qldpc15} Babar, Z., Botsinis, P., Alanis, D., Ng, S. X. \& Hanzo, L. Fifteen Years of Quantum LDPC Coding and Improved Decoding Strategies. \emph{IEEE Access} \textbf{3,} 2492-2519 (2015).

\bibitem{patrick} Fuentes, P., Etxezarreta Martinez, J., Crespo, P. M. \& Garcia-Fr\'ias, J. An approach for the construction of non-CSS LDGM-based quantum codes. \emph{Phys. Rev. A} \textbf{102}, 012423 (2020).

\bibitem{patrick2} Fuentes, P., Etxezarreta Martinez, J., Crespo, P. M., \& Garcia-Fr\'ias, J. Design of LDGM-based quantum codes for asymmetric quantum channels. \emph{Phys. Rev. A} \textbf{103,} 022617 (2021).

\bibitem{QTC} Poulin, D., Tillich, J.-P. \& Ollivier, H. Quantum Serial Turbo Codes. \emph{IEEE Trans. Inf. Theory} \textbf{55}, 2776-2798 (2009).

\bibitem{EAQTC} Wilde, M. M., Hsieh, M. \& Babar, Z. Entanglement-Assisted Quantum Turbo Codes. \emph{IEEE Trans. Inf. Theory} \textbf{60}, 1203-1222 (2014).

\bibitem{josu} Etxezarreta Martinez, J., Crespo, P. M. \& Garcia-Fr\'ias, J. Depolarizing Channel Mismatch and Estimation Protocols for Quantum Turbo Codes. \emph{Entropy} \textbf{21(12)}, 1133 (2019).

\bibitem{josu2} Etxezarreta Martinez, J., Fuentes P., Crespo,P. M. \& Garcia-Fr\'ias, J. Pauli Channel Online Estimation Protocol for Quantum Turbo Codes. \emph{2020 IEEE International Conference on Quantum Computing and Engineering (QCE)}, 102-108 (2020).

\bibitem{toric} Kitaev, A. Y. Quantum computations: Algorithms and error correction. \emph{Russian Math. Surveys} \textbf{52}, 1191-1249 (1997).

\bibitem{QEClidar} Lidar, D. \& Brun, T. Quantum Error Correction. (Cambridge Univ. Press, Cambridge, 2013).

\bibitem{surfaceFowler} Fowler, A. G., Mariantoni, M., Martinis, J. M. \& Cleland, A. N. Surface codes: Towards practical large-scale quantum computation. \emph{Phys. Rev. A} \textbf{86,} 032324 (2012).

\bibitem{fluctErrorBars} Kandala, A., Temme, K., Córcoles, A.D. et al. Error mitigation extends the computational reach of a noisy quantum processor. \emph{Nature} \textbf{567,} 491-495 (2019).

\bibitem{fluctDeph} Spring, P. A. et al. High coherence and low cross-talk in a tileable 3D integrated superconducting circuit architecture. \emph{Science Advances} \textbf{8,} 16 (2022).

\bibitem{TVQC} Etxezarreta Martinez, J., Fuentes, P., Crespo, P. M., \& Garcia-Fr\'ias, J. Time-varying Quantum Channel Models for Superconducting Qubits. \emph{npj Quantum Inf.} \textbf{7,} 115 (2021).

\bibitem{outage} Etxezarreta Martinez, J., Fuentes, P., Crespo, P. M., \& Garcia-Fr\'ias, J. Quantum outage probability for time-varying quantum channels. \emph{Phys. Rev. A} \textbf{105,} 012432 (2022).

\bibitem{TLSdefects}
Müller, C., Cole, J. H. \& Lisenfeld, J. Towards understanding two-level-systems in amorphous solids: insights from quantum circuits. \emph{Rep. Prog. Phys.} \textbf{82,} 124501 (2019).

\bibitem{TLSphysrevB} Béjanin, J. H., Earnest, C. T., Sharafeldin, A. S. \& Mariantoni, M. Interacting defects generate stochastic fluctuations in superconducting qubits. \emph{Phys. Rev. B} \textbf{104,} 094106 (2021).

\bibitem{rsaRounds} Gidney, C., \& Ekerå, M. How to factor 2048 bit RSA integers in 8 hours using 20 million noisy qubits. \emph{Quantum} \textbf{5}, 433 (2021).

\bibitem{googleSurfonemicro}Google Quantum AI. Exponential suppression of bit or phase errors with cyclic error correction. \emph{Nature} \textbf{595,} 383–387 (2021).

\bibitem{oneMicrogidney} Gidney, C. \& Fowler A. G. Flexible layout of surface code computations using AutoCCZ states. \emph{arXiv preprint} (2019).

\bibitem{IBMqexp} IBM Quantum. https://quantum-computing.ibm.com/services, 2022.

\bibitem{fading}
Biglieri, E. et al. Fading channels: information-theoretic and communications aspects. \emph{J. Phys. B: At. Mol. Opt. Phys.} \textbf{43}, 215508 (2010).

\bibitem{tse} Tse, D. \& Viswanath, P. Fundamentals of wireless communication. (Cambridge Univ. Press, Cambridge, 2005).

\bibitem{wildeQIT} Wilde, M. M. Quantum Information Theory. (Cambridge Univ. Press, Cambridge, 2017).

\bibitem{quantumcap} Gyongyosi, L., Imre, S., \& Nguyen, H. V. A Survey on Quantum Channel Capacities. \emph{IEEE Commun. Surveys Tuts.} \textbf{20,} 1149-1205 (2018).

\bibitem{bothAdd} Siddhu, V. \& Griffiths, R. B. Positivity and Nonadditivity of Quantum Capacities Using Generalized Erasure Channels. \emph{IEEE Trans. Inf. Theory} \textbf{67,} 7 (2021).

\bibitem{degenPRL} Smith, G., \& Smolin, J. A. Degenerate Quantum Codes for Pauli Channels. \emph{Phys. Rev. Lett.} \textbf{98,} 030501 (2007).

\bibitem{degen} Fuentes, P., Etxezarreta Martinez J., Crespo, P. M., \& Garcia-Fr\'ias, J. Degeneracy and its impact on the decoding of sparse quantum codes. \emph{IEEE Access}\textbf{9,} 89093-89119 (2021).

\bibitem{degen_2} Fuentes, P., Etxezarreta Martinez J., Crespo, P. M., \& Garcia-Fr\'ias, J. On the logical error rate of sparse quantum codes. \emph{IEEE Transactions on Quantum Engineering} \textbf{3,} 2100312, 1-12 (2022).

\bibitem{APDcap} Jiang, L.-Z. \& Chen, X.-Y. Quantum capacity region of simultaneous amplitude and phase damped qubit channel. \emph{International Journal of Quantum Information} \textbf{10,} 1250010 (2012).

\bibitem{degantidegAdd} Leditzky, F., Datta, N. \& Smith, G. Useful States and Entanglement Distillation. \emph{IEEE Trans. Inf. Theory} \textbf{64,} 7 (2018).

\bibitem{antiantiAdd} Holevo, A.S. Entanglement-breaking channels in infinite dimensions. \emph{Probl. Inf. Transm.} \textbf{44,} 171–184 (2008).

\bibitem{gamma1} Martinis, J.M. Saving superconducting quantum processors from decay and correlated errors generated by gamma and cosmic rays. \emph{npj Quantum Inf.} \textbf{7,} 90 (2021).

\bibitem{gamma2} Wilen, C.D., Abdullah, S., Kurinsky, N.A. et al. Correlated charge noise and relaxation errors in superconducting qubits. \emph{Nature} \textbf{594,} 369–373 (2021).

\bibitem{gamma3} McEwen, M., Faoro, L., Arya, K. et al. Resolving catastrophic error bursts from cosmic rays in large arrays of superconducting qubits. \emph{Nat. Phys.} \textbf{18,} 107–111 (2022).

\bibitem{PearsonBounds} Profillidis, V. A. \& Botzoris G. N. Chapter 5 - Statistical Methods for Transport Demand Modeling. \emph{Elsevier} 163-224 (2019).

\bibitem{strengthPear} Martínez-Gómez, E., Richards, M. T. \&  Richards D. St. P. Distance correlation methods for discovering associations in large atrophysical databases. \emph{The astrophysical journal} \textbf{781,} 39 (2014).

\bibitem{freedman}
Freedman, D., Pisani, R. \& Purves. R. Statistics. \emph{W.W. Norton} (1998).

\bibitem{bootstrap} Efron, E. \& Tibshirani, R. J. An Introduction to the Bootstrap. (Chapman \& Hall/CRC, Boca Raton, 1993).

\bibitem{CIhypothesis} Sim, J. \& Reid, N. Statistical Inference by Confidence Intervals: Issues of Interpretation and Utilization. \emph{Physical Therapy} \textbf{79,} 2 (1999).

\bibitem{classicalCorr1} Clerckx, B. \& Oestges, C. Mimo Wireless Networks (Second Edition). (Academic Press, Oxford, 2013).

\bibitem{classicalCorr2} Bhattacharya, A. \& Vaughan, R. The Antenna Correlation Coefficient in Wireless Sensor Networks. \emph{2017 XXXIInd General Assembly and Scientific Symposium of the International Union of Radio Science (URSI GASS)} 1-4 (2017).

\bibitem{classicalCorr3} Chiani, M., Win, M. Z. \& Zanella, A. On the capacity of spatially correlated MIMO Rayleigh-fading channels. \emph{IEEE Trans. Inf. Theory} \textbf{49,} 10 (2003).

\bibitem{surface}
Tuckett, D. K. qecsim - Quantum Error Correction Simulator. https://qecsim.github.io (2020).

\bibitem{experiments} Brand, D., Sinayskiy, I. \& Petruccione, F. Markovian Modelling and Calibration of IBMQ Transmon Qubits. \emph{arXiv preprint} (2022).

\end{thebibliography}
\end{document}